\documentclass[12pt, draftclsnofoot, onecolumn]{IEEEtran}
\usepackage{epsfig,latexsym}
\usepackage{float}
\usepackage{indentfirst}
\usepackage{amsmath}
\usepackage{amssymb}
\usepackage{times}
\usepackage{subfigure}
\usepackage{psfrag}
\usepackage{hyperref}
\usepackage{cite}
\usepackage{lastpage}
\usepackage{fancyhdr}
\usepackage{color}
 \usepackage{amsthm}
\usepackage{bigints}
\def\diag{\textrm{diag}}

\newtheorem{Proposition}{Proposition}
\newtheorem{Lemma}{Lemma}
\newtheorem{Corollary}{Corollary}
\newtheorem{lemma}[Lemma]{$\mathbf{Lemma}$}
\newtheorem{proposition}[Proposition]{Proposition}
\newtheorem{corollary}[Corollary]{$\mathbf{Corollary}$}

\begin{document}
\title{ {\huge OTFS-NOMA: An Efficient Approach for Exploiting Heterogenous User Mobility Profiles}}

\author{ Zhiguo Ding, Robert Schober,  Pingzhi Fan,     and H. Vincent Poor, \thanks{ 
  
\vspace{-2em}
    Z. Ding and H. V. Poor are  with the Department of
Electrical Engineering, Princeton University, Princeton, NJ 08544,
USA. Z. Ding
 is also  with the School of
Electrical and Electronic Engineering, the University of Manchester, Manchester, UK (email: \href{mailto:zhiguo.ding@manchester.ac.uk}{zhiguo.ding@manchester.ac.uk}, \href{mailto:poor@princeton.edu}{poor@princeton.edu}).
R. Schober is with the Institute for Digital Communications,
Friedrich-Alexander-University Erlangen-Nurnberg (FAU), Germany (email: \href{mailto:robert.schober@fau.de}{robert.schober@fau.de}).
P. Fan is with the Institute of Mobile
Communications, Southwest Jiaotong University, Chengdu, China (email: \href{mailto:pingzhifan@foxmail.com}{pingzhifan@foxmail.com}).

  }\vspace{-4em}}
 \maketitle
\begin{abstract}\vspace{-1em}
This paper considers a challenging communication scenario, in which   users have  heterogenous mobility profiles, e.g., some users are moving at high speeds and some users  are static.  A new non-orthogonal multiple-access (NOMA) transmission protocol that incorporates   orthogonal time frequency space (OTFS) modulation is proposed. Thereby, users with different mobility profiles are grouped together for the implementation of NOMA.  The proposed OTFS-NOMA protocol is shown to be applicable to both uplink and downlink transmission, where sophisticated transmit and receive strategies  are developed to remove inter-symbol interference and  harvest both multi-path and multi-user diversity.   Analytical results  demonstrate that both the high-mobility and low-mobility users benefit from the application of OTFS-NOMA. In particular, the use of NOMA allows  the spreading of  the high-mobility users' signals over a large amount of time-frequency resources, which enhances the OTFS resolution and improves the detection reliability. In addition, OTFS-NOMA ensures that   low-mobility users have access to   bandwidth resources which  in conventional OTFS-orthogonal multiple access (OTFS-NOMA) would be solely occupied by the high-mobility users. Thus, OTFS-NOMA improves the spectral efficiency and reduces latency.  


\end{abstract} \vspace{-2em}
\section{Introduction}
Non-orthogonal multiple access (NOMA) has been recognized as a paradigm  shift for the design of multiple access techniques for the next generation wireless networks \cite{nomama,8010756,7676258}. Many existing works on NOMA have focused on scenarios with low-mobility users, where users with different channel conditions or quality of service (QoS) requirements are grouped together for the implementation of NOMA.   For example, in  power-domain NOMA,   a base station serves two users simultaneously \cite{NOMAPIMRC,Nomading}. In particular, the base station first orders the users according to their channel conditions, where the `weak user' which has a poorer connection to the base station is generally allocated more transmission power and the other user, referred to as the `strong user', is allocated less power. As such, the two users can be served in the same time-frequency resource,  which improves the spectral efficiency compared to orthogonal multiple access (OMA). In the case that users have similar channel conditions, grouping users with different QoS requirements can  facilitate the   implementation  of NOMA and  effectively  exploit  the   potential of NOMA    \cite{Zhiguo_CRconoma, 8541121,8611381}.  Various existing studies have shown that   the NOMA    principle can be applied to different  communication networks, such as millimeter-wave networks \cite{Zhiguo_mmwave,8493528 }, massive multiple-input multiple-output (MIMO) systems \cite{Zhiguo_massive,8598847}, visible light communication networks \cite{7572968,7275086}, and mobile edge computing \cite{MECding}. 
 
This paper considers the application of NOMA to a challenging communication scenario, where users have heterogeneous    mobility profiles. Different from the existing works in \cite{8246842,7976286}, the use of     orthogonal time frequency space (OTFS) modulation is considered in this paper because of its superior performance    in scenarios with doubly-dispersive   channels \cite{OTFS, OTFS2,8503182}.  Recall that the key idea of OTFS is to use the delay-Doppler plane, where users' signals    are orthogonally placed. Compared to conventional modulation schemes, such as orthogonal frequency-division multiplexing  (OFDM), OTFS offers the benefit that the time-invariant channel gains   in the delay-Doppler plane can be utilized, which simplifies channel estimation and signal detection in high-mobility scenarios. The impact of  pulse-shaping waveforms on the performance of OTFS was studied  in \cite{8516353}, and the design of interference cancellation and iterative detection for OTFS was investigated in \cite{8424569}. The diversity achieved by OTFS was  studied in \cite{otfs_diversity}, and the application of OTFS to multiple access was proposed in \cite{8515088}. In  \cite{8422467} and \cite{8647394},  the concept of OTFS was combined with MIMO, which showed that the use of spatial degrees of freedom can further enhance the performance of OTFS.

This paper considers the application of OTFS to NOMA communication networks, where  the coexistence of NOMA and OTFS is investigated. In particular, this paper makes  the following contributions:  
 
1) A spectrally efficient OTFS-NOMA transmission protocol  is proposed by grouping users with different mobility profiles for the implementation of NOMA. {\it On the one hand}, users with high mobility are served in the delay-Doppler plane, and their signals are modulated  by OTFS. {\it On the other hand},  users with low mobility are served in the time-frequency plane, and their signals are modulated in a manner  similar to conventional OFDM.

2) The proposed new OTFS-NOMA protocol is applied to both uplink and downlink transmission, where different rate and power allocation   policies  are used  to suppress   multiple access interference. 
In addition, sophisticated equalization techniques, such as the frequency-domain zero-forcing linear equalizer (FD-LE) and the decision feedback equalizer (FD-DFE), are proposed to remove the   inter-symbol interference  in the delay-Doppler plane. The impact of the developed equalization techniques on the performance of OTFS-NOMA is   analyzed by using the outage probability as the criterion.  Strategies to harvest   multi-path diversity and     multi-user diversity are also  introduced, which can further improve the outage performance of OTFS-NOMA transmission.

3)
The developed analytical results demonstrate that both the high-mobility and the low-mobility users benefit from the proposed    OTFS-NOMA scheme.  The use of NOMA allows     the high-mobility users' signals to be spread over   a large amount of time-frequency resources without degrading the spectral efficiency.  As a result, the OTFS resolution, which determines whether the users' channels can be accurately located in the delay-Doppler plane, is enhanced significantly, and  therefore, the reliability of detecting the high-mobility users' signals is     improved.  We note that, in  OTFS-OMA, enhancing the OTFS resolution  implies  that   a large amount of time and frequency resources are solely occupied by the high-mobility users, which reduces the overall spectral efficiency since the high-mobility users' channel conditions are typically weaker than those of the low-mobility users. 
In contrast, the use of OTFS-NOMA ensures that the low-mobility users can access the bandwidth resources which would be solely occupied by the high-mobility users in the OMA mode. Hence, OTFS-NOMA improves   spectral efficiency and reduces latency. In addition, we note that  for the low-mobility users, using   OFDM  yields the same reception reliability   as  using OTFS, as pointed out in \cite{8599041}. Therefore, the proposed  OTFS-NOMA scheme, which serves the low-mobility users in the time-frequency plane and modulates the low-mobility users' signals in a manner similar to OFDM, offers the same  reception reliability  as OTFS-OMA, which serves the low-mobility users in the delay-Doppler plane and modulates the low-mobility users' signals by OTFS. However, OTFS-NOMA has the benefit of reduced system complexity because the use of the complicated OTFS  transforms is avoided. 
 


\vspace{-1em}
 \section{ Basics of OTFS-NOMA } \label{section 2}

\subsection{Time-Frequency Plane and Delay-Doppler Plane}
The key idea of OTFS-NOMA is to efficiently use both the time-frequency plane and the delay-Doppler plane. A discrete time-frequency plane  is obtained by sampling  at intervals of $T$ s and $\Delta f$ Hz as follows: 
\begin{align}
\Lambda_{\text{TF}} = \{(nT, m\Delta f), n =0, \cdots, N-1, m =0, \cdots, M-1\}, 
\end{align}
and the corresponding discrete delay-Doppler plane  is given by
\begin{align}
\Lambda_{\text{DD}} = \left\{\left(\frac{k}{NT}, \frac{l}{M\Delta f}\right), k =0, \cdots, N-1, l =0, \cdots, M-1\right\}.
\end{align}
The choices for $T$ and $\Delta f$ are determined by the channel characteristics, as explained  in the following subsection.   

\vspace{-1em}

\subsection{Channel Model}
This paper considers a multi-user communication network in which  one base station communicates with $(K+1)$   users, denoted by $\text{U}_i$, $0\leq i\leq K$. 
Denote    $\text{U}_i$'s channel response in the delay-Doppler plane by $ h_i(\tau, \nu)$, where $\tau$ denotes the delay and $\nu$ denotes the Doppler shift. OTFS uses the sparsity feature of a wireless channel in the delay-Doppler plane, i.e., there are a small number of propagation paths between a transmitter and a receiver \cite{OTFS, OTFS2,8424569}, which means that $ h_i(\tau, \nu)$ can be expressed as follows:
\begin{align}\label{dd channel}
 h_i(\tau, \nu) = \sum^{P_i}_{p=0}h_{i,p}\delta(\tau-\tau_{i,p})\delta(\nu-\nu_{i,p}),
\end{align}
where $(P_i+1)$ denotes the number of propagation paths, and $h_{i,p}$, $\tau_{i,p}$, and $\nu_{i,p}$ denote the complex Gaussian channel gain, the delay, and the Doppler shift associated with  the $p$-th propagation path. We assume that the $h_{i,p}$, $0\leq p\leq P_i$, are independent and identically distributed (i.i.d.)  random variables\footnote{In order to simplify the performance analysis, we assume that the users' channels are  i.i.d. In practice, it is likely  that the high-mobility users' channel conditions are worse than the low-mobility users' channel conditions. This channel difference is beneficial for the implementation of NOMA, and hence can further increase the performance gain of OTFS-NOMA over OTFS-OMA.  },  i.e., $h_{i,p}\sim CN\left(0,\frac{1}{P_i+1}\right)$, which means $\sum^{P_i}_{p=0}\mathcal{E}\{|h_{i,p}|^2\}=1$,  where $\mathcal{E}\left\{\cdot\right\}$ denotes the expectation operation.  Hence,    The discrete   delay and Doppler tap indices  for the $p$-th path of $h_{i}(\tau,\nu)$, denoted by $l_{\tau_{i,p}}$ and $k_{\nu_{i,p}}$, are given by
\begin{align}
\tau_{i,p} =\frac{l_{\tau_{i,p}}+\tilde{l}_{\tau_{i,p}}}{M\Delta f}, \quad \nu_{i,p} =\frac{k_{\nu_{i,p}}+\tilde{k}_{\nu_{i,p}}}{NT},
\end{align}
where  $\tilde{l}_{\tau_{i,p}}$ and $\tilde{k}_{\nu_{i,p}}$ denote the fractional delay and the fractional Doppler shift, respectively. 

The construction of $\Lambda_{\text{TF}} $  and $\Lambda_{\text{DD}} $  needs to ensure that  $T$ is not smaller than the maximal delay spread, and $\Delta f$ is not smaller than the largest Doppler shift,  i.e., $T\geq \max\{\tau_{i,p},0\leq p\leq P_i, 0\leq i\leq K\}$ and $\Delta f\geq \max\{\nu_{i,p},0\leq p\leq P_i, 0\leq i\leq K\}$. In addition, the choices of $N$ and $M$ affect   the OTFS resolution, which determines  whether $h_i(\tau, \nu)$   can be accurately located in the discrete delay-Doppler plane. In particular, $M$ and $N$ need to be sufficiently large to approximately achieve   ideal OTFS resolution,  which ensures   that $\tilde{l}_{\tau_{i,p}}=\tilde{k}_{\nu_{i,p}}=0$, such that interference caused by   fractional delay and fractional Doppler shift is effectively suppressed~\cite{OTFS}.  
\vspace{-1em}

\subsection{General Principle of OTFS-NOMA}
The general principle  of the proposed OTFS-NOMA scheme is to utilize both the delay-Doppler plane and the time-frequency plane, where users with heterogenous mobility profiles are grouped together and served simultaneously. {\it On the one hand}, for the users with high mobility, their signals are placed in the delay-Doppler plane, which  means that  the time-invariant  channel gains in the delay-Doppler plane can be exploited.  It is worth pointing out that in order to ensure that the channels  can be located in the delay-Doppler plane, both $N$ and $M$ need to be large, which is a disadvantage of OTFS-OMA, since a significant number of frequency channels (e.g., $M\Delta f$) are   occupied for a long time (e.g., $NT$) by the high-mobility users whose channel conditions can be quite  weak. The use of OTFS-NOMA facilitates  spectrum  sharing and hence ensures that the high-mobility users' signals can be spread over a large amount of time-frequency resources without degrading the spectral efficiency. 

{\it On the other hand}, for the users with low mobility, their signals are placed in the time-frequency plane. The interference between the users with different mobility profiles is managed by using the principle of NOMA. As a result, OTFS-NOMA improves the overall spectral efficiency since it avoids that the bandwidth resources are solely occupied by the high-mobility users as in OTFS-OMA. In addition, the complexity of detecting the low-mobility users' signals is reduced, compared to OTFS-OMA which serves  all users  in the delay-Doppler plane.  


In this paper, we assume that, among the $(K+1)$ users, $\text{U}_0$  is a user with   high mobility, and the remaining $K$ users, $\text{U}_i$ for $1\leq i\leq K$, are low-mobility  users, which are  referred  to as   `NOMA' users\footnote{Alternatively,   multiple high-mobility users can be served in the delay-Doppler plane,  and this change  has no impact on the downlink  results obtained in this paper. For the uplink case, the results developed in the paper are applicable to the case with multiple high-mobility users if  adaptive data rate transmission is employed.   }. For OTFS-OMA,  we assume that $\text{U}_0$ solely occupies all   $NM$ resource blocks in $\Lambda_{\text{DD}}$. In OTFS-NOMA,   $\text{U}_i$, for $1\leq i\leq K$, are opportunistic NOMA users and their  signals are placed in $\Lambda_{\text{TF}}$.   The design  of downlink OTFS-NOMA  transmission will be discussed in detail in  Sections \ref{section 3},    \ref{section 4},  and \ref{section 5}.   The application of OTFS-NOMA for uplink transmission will be considered   in Section \ref{section 6} only briefly,  due to space limitations.    

 

\section{Downlink OFTS-NOMA  - System Model } \label{section 3}
In this section, the  OTFS-NOMA downlink transmission protocol is described. 
In particular, assume  that the base station sends  $N M$   signals to $\text{U}_0$, denoted by $x_0[k,l]$, $k\in\{0, \cdots, N-1\}$, $l\in\{0,\cdots, M-1\}$.   By using the inverse symplectic finite Fourier transform (ISFFT),    the high-mobility user's symbols placed in the delay-Doppler plane are converted to   $N M$ symbols in the time-frequency plane as follows \cite{OTFS}:
\begin{align}
X_0[n,m] =& \frac{1}{NM}\sum^{N-1}_{k=0}\sum^{M-1}_{l=0}x_0[k,l]e^{j2\pi \left(\frac{kn}{N}-\frac{ml}{M}\right)}
,
\end{align}
where   $n\in\{0, \cdots, N-1\}$ and $m\in\{0,\cdots, M-1\}$. We note that the $N M$  time-frequency signals can be viewed as   $N$ OFDM symbols containing $M$ signals each.   We assume that a rectangular window is applied to the transmitted and received signals.  

The NOMA users' signals are placed directly in the time-frequency plane, and are superimposed with the high-mobility user's   signals, $X_0[n,m]$. With $NM$ orthogonal resource blocks available in the time-frequency plane,   there are different ways for the $K$ users to share the resource blocks.  For illustration purposes, we assume that $M$ users are selected from the $K$ opportunistic NOMA users, where  each NOMA user is to occupy one frequency subchannel and receive $N$ information bearing symbols, denoted by   $x_i(n)$, for $1\leq i \leq M$ and $0\leq n\leq N-1$.   The criterion for user scheduling and its impact on the performance of OTFS-NOMA will be discussed  in Section \ref{section 5}.  
Denote  the time-frequency signals to be sent to $\text{U}_i$ by $X_i[n,m]$, $1\leq i \leq N$. The following mapping scheme is used in this paper\footnote{We note that   mapping schemes different from \eqref{indooruser2} can also be used. For example, if $N$ users are scheduled and each user is to occupy one time slot and receives an OFDM-like symbol containing $M$ signals, we can set $ X_i[n,m] = x_i(m)$, for $n=i-1$.    }: 
\begin{align}\label{indooruser2}
X_i[n,m] =\left\{\begin{array}{ll} x_i(n) & \text{if}\quad m=i-1 \\ 0 &\text{otherwise}\end{array}\right.,
\end{align}
for $1\leq i \leq M$ and $0\leq n \leq N-1$.  

The base station superimposes $\text{U}_0$'s time-frequency signals with the NOMA users' as follows: 
\begin{align}\label{downlink1}
X[n,m] =& \frac{\gamma_0}{NM}\sum^{N-1}_{k=0}\sum^{M-1}_{l=0} {x}_0[k,l]e^{j2\pi \left(\frac{kn}{N}-\frac{ml}{M}\right)} +\sum^{M}_{i=1}\gamma_i X_i[n,m]
,
\end{align}
where   $\gamma_i$ denotes the NOMA power allocation coefficient of user $i$, and   $\sum^{M}_{i=0}\gamma_i^2=1$. 

The  transmitted signal at  the base station is obtained by applying   the Heisenberg transform to $X[n,m]$. By assuming      perfect orthogonality  between the transmit and receive pulses,   the received signal  at $\text{U}_i$  in the time-frequency plane  can be modelled  as follows \cite{OTFS, OTFS2,8424569} : 
\begin{align}\label{downlink1x}
Y_i[n,m] =&  H_i[n,m]X[n,m] +W_i[n,m],
\end{align}
where $W_i(n,m)$ is the white Gaussian noise in  the time-frequency plane, and $H_i(n,m) = \int\int h_i(\tau, \nu) e^{j2\pi \nu nT} e^{-j2\pi(\nu+m\Delta f)\tau} d\tau d\nu$.

\section{Downlink OTFS-NOMA   - Detecting   the High-Mobility User's Signals} \label{section 4}
For the proposed downlink OTFS-NOMA scheme,  $\text{U}_0$ directly detects its  signals in the delay-Doppler plane by treating the NOMA users' signals   as noise. In particular, in order to detect $\text{U}_0$'s signals,   the symplectic finite Fourier transform (SFFT) is applied to $Y_0[n,m]$ to obtain the delay-Doppler estimates as follows:
\begin{align}
y_0[k,l] =&\frac{1}{NM}\sum^{N-1}_{n=0}\sum^{M-1}_{m=0}  {Y}_0[n,m] e^{-j2\pi \left(\frac{nk}{N}-\frac{ml}{M}\right)}\\\nonumber
=&\frac{1}{NM}\sum^{M}_{q=0}\gamma_q\sum^{N-1}_{n=0}\sum^{M-1}_{m=0} x_q[n,m]    h_{w,0}\left(\frac{k-n}{NT}, \frac{l-m}{M\Delta f}\right) +z_0[k,l],
\end{align}
where $z_0[k,l]$ is     complex Gaussian noise, $ {x}_q[k,l]$, $1\leq q\leq M$, denotes the delay-Doppler representation of $X_q[n,m]$ and can be obtained by applying the SFFT to $X_q[n,m]$, 
the channel $h_{w,0}(\nu', \tau')$ is given by
\begin{align}
h_{w,0}(\nu', \tau') = \int\int h_i(\tau, \nu) w(\nu'-\nu, \tau'-\tau) e^{-j2\pi \nu \tau} d\tau d\nu,
\end{align}
and $w(\nu, \tau) = \sum^{N-1}_{c=0}\sum^{M-1}_{d=0} e^{-j2\pi (\nu c T - \tau d \Delta f)}$. To simplify the analysis,  the power of the complex-Gaussian distributed noise  is assumed to be normalized, i.e., $z_i[k,l]\sim CN(0,1)$, where $CN(a,b)$ denotes a complex Gaussian distributed random variable  with mean $a$ and variance $b$.  

By applying the channel model in \eqref{dd channel}, the relationship between the transmitted signals and the observations in the delay-Doppler plane can be expressed as follows \cite{OTFS, OTFS2,8424569} : 
\begin{align}\label{yx modelx1}
y_0[k,l]  
=&\sum^{M}_{q=0} \gamma_q\sum^{P_0}_{p=1} h_{0,p} x_q[(k-k_{\nu_{0,p}})_N, (l-l_{\tau_{0,p}})_M]  +z_0[k,l],
\end{align}
where  $(\cdot)_N$ defines the modulo $N$ operator.  As in \cite{8422467,otfs_diversity,8515088}, we assume that $N$ and $M$ are sufficiently large to ensure that  both $\tilde{k}_{\nu_{0,p}}$ and $\tilde{l}_{\tau_{0,p}}$ are zero, i.e., there is no interference caused by fractional delay or fractional Doppler shift. We note that for OTFS-OMA, increasing $N$ and $M$ can significantly reduce spectral efficiency, whereas the use of large $N$ and $M$ becomes possible for OTFS-NOMA   because of the spectrum  sharing of  users with different mobility profiles.

Define $\mathbf{y}_{0,k}=\begin{bmatrix}y_0[k,0] & \cdots & y_0[k,M-1]\end{bmatrix}^T$ and $\mathbf{y}_0=\begin{bmatrix} \mathbf{y}_{0,0}^T&\cdots &\mathbf{y}_{0,N-1}^T\end{bmatrix}^T$. Similarly, the signal vector $ {\mathbf{x}}_i$ and the noise vector $\mathbf{z}_0$ are constructed from $x_i[k,l]$ and $z_0[k,l]$, respectively.  Based on \eqref{yx modelx1}, the system  model can be expressed in  matrix form as follows:
\begin{align}\label{yx model2}
 \mathbf{y}_0 =\gamma_0\mathbf{H}_0 \mathbf{x}_0 +\underset{\text{Interference and noise terms}}{\underbrace{\sum^{M}_{q=1}\gamma_q\mathbf{H}_0 {\mathbf{x}}_q +\mathbf{z}_0}},
\end{align}
where   $\mathbf{H}_0$ is a block-circulant matrix and defined as follows:
\begin{align}
\mathbf{H}_0 = \begin{bmatrix} \mathbf{A}_{0,0} & \mathbf{A}_{0,N-1}&\cdots &\mathbf{A}_{0,2} & \mathbf{A}_{0,1} \\  
\mathbf{A}_{0,1} &\mathbf{A}_{0,0}&\ddots &\mathbf{A}_{0,3} &\mathbf{A}_{0,2}
\\
\vdots&\ddots&\ddots &\ddots &\vdots
\\
\mathbf{A}_{0,N-2} &\mathbf{A}_{0,N-3}&\ddots &\mathbf{A}_{0,0} &\mathbf{A}_{0,N-1}
\\
\mathbf{A}_{0,N-1} &\mathbf{A}_{0,N-2}&\ddots &\mathbf{A}_{0,1} &\mathbf{A}_{0,0}
\end{bmatrix},
\end{align}
and each submatrix $\mathbf{A}_{0,n}$ is an $M\times M$ circulant matrix whose structure is determined  by \eqref{yx modelx1}.

{\it Example:} Consider a special case with $N=4$ and $M=3$, and $\text{U}_0$'s channel is given by
\begin{align}\label{yx channel model}
h_0(\tau, \nu)  
=& h_{0,0}\delta (\tau) \delta (\nu) + h_{0,1}\delta\left(\tau -\frac{1}{M\Delta f} \right)\delta\left(\nu - \frac{3}{NT} \right) ,
\end{align}
which means $
k_0=0$, $k_1 = 3$, $l_0=0$, $l_1 = 1$. 
Therefore, the block-circulant matrix is given by 
\begin{align}
\mathbf{H}_0 = \begin{bmatrix} \mathbf{A}_{0,0} & \mathbf{A}_{0,3} &\mathbf{A}_{0,2} & \mathbf{A}_{0,1} \\  
\mathbf{A}_{0,1} &\mathbf{A}_{0,0} &\mathbf{A}_{0,3} &\mathbf{A}_{0,2}
\\
\mathbf{A}_{0,2} &\mathbf{A}_{0,1} &\mathbf{A}_{0,0} &\mathbf{A}_{0,3}
\\
\mathbf{A}_{0,3} &\mathbf{A}_{0,2} &\mathbf{A}_{0,1} &\mathbf{A}_{0,0}
\end{bmatrix},
\end{align}
where $\mathbf{A}_{0,0}=h_{0,0}\mathbf{I}_3$, $\mathbf{A}_{0,2}=\mathbf{A}_{0,3}=\mathbf{0}_{3\times 3}$ and 
$
\mathbf{A}_{0,1} = \begin{bmatrix} 0 &0 & h_{0,1} \\ h_{0,1} &0 &0\\ 0 &h_{0,1} & 0\end{bmatrix}$.

{\it Remark 1:} It is well known that an $n\times n$ circulant matrix can be diagonalized by the $n\times n$ fast Fourier transform (FFT) and inverse FFT matrices, denoted by $\mathbf{F}_n$ and $\mathbf{F}_n^{-1}$, respectively.  We note that directly  applying the FFT factorization to $\mathbf{H}_0$ is not possible, since $\mathbf{H}_0$ is not a circulant  matrix, but a block circulant matrix.  

Because  of the structure of $\mathbf{H}_0$, inter-symbol interference  still exists in the considered OTFS-NOMA system, and equalization is needed.     We consider  two equalization approaches, FD-LE and FD-DFE, which were both originally developed  for single-carrier transmission  with cyclic prefix \cite{995852,1195504}. 

\vspace{-1em}

\subsection{  Design and Performance  of FD-LE}
The proposed FD-LE consists of two steps. The first step is to multiply the observation vector $\mathbf{y}_0$ by $\mathbf{F}_N \otimes \mathbf{F}_M^H  $, which leads to the result  in the following proposition. 
\begin{proposition}\label{proposition1}
By applying the detection matrix $\mathbf{F}_N \otimes \mathbf{F}_M^H  $ to   observation vector $\mathbf{y}_0$,  the received signals  for OTFS-NOMA downlink transmission can be written as follows:
\begin{align}\label{yx model6}
\tilde{\mathbf{y}}_0
 = & \mathbf{D}_0  (\mathbf{F}_N \otimes \mathbf{F}_M^H)\left(\gamma_0\mathbf{x}_0 +\sum^{M}_{q=1}\gamma_q  {\mathbf{x}}_q \right) +\tilde{\mathbf{z}}_0,
\end{align} 
where $\tilde{\mathbf{y}}_0=(\mathbf{F}_N \otimes \mathbf{F}_M^H  ) \mathbf{y}_0 $, $\tilde{\mathbf{z}}_0=(\mathbf{F}_N\otimes \mathbf{F}_M^H) \mathbf{z}_0$, $\mathbf{D}_0$ is a diagonal matrix whose $(kM+l+1)$-th  diagonal element is given by  
\begin{align}
D_{0}^{k,l}= \sum^{N-1}_{n=0}\sum^{M-1}_{m=0} a_{0,n}^{m,1} e^{j2\pi \frac{lm}{M}} e^{-j2\pi \frac{kn}{N}},
\end{align}
for $0\leq k \leq N-1$, $0\leq l \leq M-1$, and 
 $a^{m,1}_{0,n}$ is the element located in the $(nM+m+1)$-th row and the first column of $\mathbf{H}_0$. 
\end{proposition} 
\begin{proof}
Please refer to Appendix \ref{proof1}. 
\end{proof}

With the simplified signal  model shown in \eqref{yx model6}, the second step of FD-LE is to apply $\left( \mathbf{F}_N \otimes \mathbf{F}_M^H\right)^{-1}\mathbf{D}_0 ^{-1}$ to $\tilde{\mathbf{y}}_0$. Thus,   $\text{U}_0$'s received signal is given by
\begin{align}\label{approach 1 system model 1}
\breve{\mathbf{y}}_0
 = &  \gamma_0\mathbf{x}_0 +\underset{\text{Interference and noise terms}}{\underbrace{\sum^{M}_{q=1} \gamma_q {\mathbf{x}}_q  +\left( \mathbf{F}_N \otimes \mathbf{F}_M^H\right)^{-1}\mathbf{D}_0 ^{-1}\tilde{\mathbf{z}}_0}},
\end{align} 
where $\breve{\mathbf{y}}_0 = \left( \mathbf{F}_N \otimes \mathbf{F}_M^H\right)^{-1}\mathbf{D}_0 ^{-1}\tilde{\mathbf{y}}_0$. To simplify the  analysis,  we assume that the powers of all users' information-bearing signals are identical, which means that the transmit signal-to-noise ratio (SNR) can be defined as $\rho=\mathcal{E}\{|x_0[k,l]|^2 \}=\mathcal{E}\{|x_i(n)|^2 \}$, since the noise power is assumed to be normalized \footnote{Following   steps similar to those in the proofs for Proposition \ref{proposition1}, $W_i[n,m]\sim CN(0,1)$ if $z_i[k,l]\sim CN(0,1)$. }.       The following lemma provides  the   signal-to-interference-plus-noise ratio (SINR) achieved by FD-LE.  
\begin{lemma} \label{lemma1}
Assume that $\gamma_i=\gamma_1$, for $1\leq i \leq N$. By using FD-LE, the SINRs for detecting all $x_0[k,l]$,  $0\leq k \leq N-1$ and $0\leq l \leq M-1$, are identical  and   given by 
\begin{align}\label{sinrzf}
\text{SINR}_{0,kl}^{\text{LE}} = \frac{\rho\gamma_0^2}{\rho \gamma^2_1 +\frac{1}{NM}\sum^{N-1}_{\tilde{k}=0}\sum^{M-1}_{\tilde{l}=0} |D_0^{\tilde{k},\tilde{l}}|^{-2}}.
\end{align}

\end{lemma}
\begin{proof}
Please refer to Appendix \ref{prooflemma1}. 
\end{proof}
{\it Remark 2:} The proof of Lemma \ref{lemma1} shows that $\sum^{M}_{i=0}\gamma_i^2=1$ can be simplified as $\gamma_0^2+\gamma_i^2=1$ for $1\leq i \leq M$, which is the motivation  for assuming  $\gamma_i=\gamma_1$. Following steps similar to those in the proofs for Proposition \ref{proposition1} and Lemma \ref{lemma1}, one can show that   directly applying $\mathbf{H}_0^{-1}$ to the observation vector yields   the same SINR. However, the proposed   FD-LE scheme can be implemented  more efficiently since  $ \left( \mathbf{F}_N \otimes \mathbf{F}_M^H\right)^{-1} =  \mathbf{F}_N^H \otimes \mathbf{F}_M  $  and $\mathbf{D}_0$ is a diagonal matrix. Hence,  the inversion of a full $NM\times NM$ matrix is avoided. 

  The outage probability achieved by FD-LE is given by $\mathrm{P}(\log(1+\text{SINR}_{0,kl}^{\text{LE}} )<R_0)$, where $R_i$, $0\leq i\leq M$, denotes $\text{U}_i$'s target data rate.  
 It is difficult to analyze the outage probability for  the  following two reasons. First, the $D_0^{k,l}$, $k\in \{0, \cdots, N-1\}$,   $l\in\{0, \cdots, M-1\}$, are not statistically independent, and second,   the distribution of a sum of the inverse of exponentially distributed random variables is difficult to characterize. The following lemma provides an asymptotic result  for the outage probability based on the SINR provided in Lemma \ref{lemma1}.

\begin{lemma}\label{lemmax2}
If $\gamma_0^2 > \gamma^2_1\epsilon_0$, the diversity order achieved by FD-LE is one, where $\epsilon_0=2^{R_0}-1$. Otherwise, the outage probability is always one. 
\end{lemma}
\begin{proof}
Please refer to Appendix \ref{plemmax2}.
\end{proof}

{\it Remark 3:} Recall that the diversity order achieved by OTFS-OMA, where  the high-mobility user, $\text{U}_0$,   solely occupies  the bandwidth resources, is also one. Therefore, introducing the low-mobility users' signals in the time-frequency plane via OTFS-NOMA  does not compromise  $\text{U}_0$'s diversity order, but  improves the spectral efficiency, compared to OTFS-OMA. 
 
\vspace{-1em}

\subsection{Design and Performance  of FD-DFE}\label{subsection fd-dfe}
Different from FD-LE, which is a  linear equalizer,  FD-DFE is based on the idea of feeding back previously detected symbols. Since both $\mathbf{x}_0$ and $\mathbf{x}_q$, $q\geq 1$, experience the same fading channel, we first define $\mathbf{x}= \gamma_0   {\mathbf{x}}_0 + \sum^{M}_{q=1}\gamma_q   {\mathbf{x}}_q$, which are the signals to be recovered by FD-DFE.  Given the observations shown  in \eqref{yx model2}, the outputs  of the FD-DFE  are given by 
\begin{align}\label{yx model approach II}
\hat{\mathbf{x}}   =\mathbf{P}_0   \mathbf{y}_0  -  \mathbf{G}_0 \check{\mathbf{x}},
\end{align}
where $\check{\mathbf{x}}$ denote the decisions made on the symbols $ {\mathbf{x}}$,  $\mathbf{P}_0$ is the feed-forward part of the equalizer, and $\mathbf{G}_0$ is the feedback part of the equalizer.   Similar to \cite{995852,1195504}, we use  the following choices for $\mathbf{P}_0$ and $\mathbf{G}_0$: $
\mathbf{P}_0 = \mathbf{L}_0 (\mathbf{H}_0^H\mathbf{H}_0)^{-1} \mathbf{H}_0^H,
 $
and 
 $
\mathbf{G}_0 = \mathbf{L}_0 - \mathbf{I}_{NM}
 $, 
where $\mathbf{L}_0$ is a lower triangular matrix with its main diagonal elements being ones in order to ensure causality of the feedback signals. 
With the above choices for $\mathbf{P}_0$ and $\mathbf{G}_0$, $\text{U}_0$'s signals can be detected as follows:
\begin{align}\label{yx model approach II 2}
 \hat{\mathbf{x}} =&\mathbf{L}_0 (\mathbf{H}_0^H\mathbf{H}_0)^{-1} \mathbf{H}^H_0  \mathbf{y}_0  -  (  \mathbf{L}_0 - \mathbf{I}_{NM}) \check{\mathbf{x}}.
\end{align}
For FD-DFE,   $\mathbf{L}_0$ is obtained   from the Cholesky decomposition of  $\mathbf{H}_0$, i.e., $\mathbf{H}_0^H\mathbf{H}_0=\mathbf{L}_0^H\mathbf{\Lambda}_0\mathbf{L}_0$, where $\mathbf{L}_0$ is the desirable  lower triangular matrix, and $\mathbf{\Lambda}_0$ is a diagonal matrix.  Therefore,  the estimates of $ {\mathbf{x}}_0$ can be rewritten   as follows:
\begin{align}\nonumber
 \hat{\mathbf{x}} 
  =& 
 \mathbf{x}    + \mathbf{L}_0 (\mathbf{H}_0^H\mathbf{H}_0)^{-1} \mathbf{H}^H_0\mathbf{z}_0    \\  
 =& 
  \label{yx model approach II 3}  
  \gamma_0\mathbf{x}_0  +\underset{\text{Interference and noise terms}}{\underbrace{    \sum^{M}_{q=1}\gamma_q {\mathbf{x}}_q  + \mathbf{L}_0 (\mathbf{H}_0^H\mathbf{H}_0)^{-1} \mathbf{H}^H_0\mathbf{z}_0   }},
\end{align}
where  perfect decision-making is assumed, i.e., $\check{\mathbf{x}}=\mathbf{x}$, and there is no error propagation    \cite{1195504,Devillersphd,cypd1}.  We note that \eqref{yx model approach II 3} yields an upper bound on the reception reliability of FD-DFE when error propagation cannot be completely avoided. 

Following   steps similar to those in the proof of Lemma \ref{lemma1},  the covariance matrix for the interference-plus-noise term can be found as follows:
\begin{align}
 \mathbf{C}_{\text{cov}} = &  \rho \gamma_1^2    \mathbf{I}_{MN}+   \mathbf{L}_0 (\mathbf{H}_0^H\mathbf{H}_0)^{-1}  \mathbf{L}_0 ^H  =\rho \gamma_1^2    \mathbf{I}_{MN}+   \mathbf{\Lambda}_0^{-1},
\end{align}
where the last step follows  from the fact that   $\mathbf{L}_0$ is obtained from the Cholesky decomposition of $\mathbf{H}_0$. Therefore,    the SINR for detecting  $x_0[k,l]$ can be expressed as follows:
\begin{align}\label{SINR approach II}
\text{SINR}_{0,kl} = \frac{\rho \gamma_0^2}{\rho\gamma_1^2 + \lambda_{0,kl}^{-1}},
\end{align}
where $\lambda_{0,kl}$ is the $(kM+l+1)$-th element on the main diagonal of $\mathbf{\Lambda}_0$.

{\it Remark 4:} We note that  there is a fundamental difference between the two equalization schemes. One can observe from \eqref{sinrzf} that   the SINRs achieved by FD-LE  for different $x_0[k,l]$ are identical. However, for FD-DFE,   different symbols experience different effective fading gains, $\lambda_{0,kl}$.   Therefore, FD-DFE can realize unequal error protection for data streams with different priorities. This comes  at the price of a higher computational complexity.  

We further note that  the use of FD-DFE also ensures that multi-path diversity can be harvested, as shown in the following.  
The outage performance analysis for FD-DFE requires   knowledge of  the   distribution of the effective channel gains, $\lambda_{0,kl}$. Because of the implicit  relationship between   $\mathbf{\Lambda}_{0}$ and $\mathbf{H}_0$, a general expression for the outage probability  achieved by FD-DFE is difficult to obtain. 
 However, analytical results can be developed for special cases to show  that  the use of FD-DFE can realize the maximal multi-path diversity. 

In particular, the SINR for $x_0[N-1,M-1]$ is a function of $\lambda_{0,(N-1)(M-1)}$ which is the last element on the main diagonal of $\mathbf{\Lambda}_0$. Recall that $\mathbf{\Lambda}_0$  is obtained via    Cholesky decomposition,  i.e., $\mathbf{H}_0^H\mathbf{H}_0=\mathbf{L}_0^H\mathbf{\Lambda}_0\mathbf{L}_0$. Because $\mathbf{L}_0$ is a lower triangular matrix, $ \lambda_{0,(N-1)(M-1)}$ is equal to the element of $\mathbf{H}_0^H\mathbf{H}_0$ located in the $NM$-th column and the $NM$-th row, which means 
\begin{align} 
\lambda_{0,(N-1)(M-1)} = \sum^{P_0}_{p=0}|h_{0,p}|^2.
\end{align}
Since    the channel gains are i.i.d. and follow $h_{0,p}\sim CN(0,\frac{1}{P_0+1})$, the probability density function (pdf) of $\sqrt{P_0+1} \lambda_{0,(N-1)(M-1)} $ is given by
\begin{align}
f(x) = \frac{1}{P_0!}e^{-x}x^{P_0}.
\end{align}
By using the above pdf, the outage probability and the diversity order   can be obtained by some algebraic manipulations, as shown in the following corollary.

\begin{corollary}\label{lemma3}
Assume    $\gamma_0^2 > \gamma^2_1\epsilon_0$.  The use of FD-DFE realizes the following outage probability for detection of $x_0[N-1,M-1]$:
\begin{align}
\mathrm{P}^0_{N-1,M-1} = \frac{1}{P_0!}g\left(P_0+1, \frac{\epsilon_0(P_0+1)}{\rho(\gamma_0^2-\gamma_1^2\epsilon_0)}\right),
\end{align} 
 where $g(\cdot)$ denotes the incomplete Gamma function.  The full multi-path diversity order, $P_0+1$, is achievable for     $x_0[N-1,M-1]$
\end{corollary} 

{\it Remark 5:}  
The results in Corollary \ref{lemma3} can be extended to OTFS-OMA with FD-DFE straightforwardly.       We also note that   diversity gains larger than one are not achievable with FD-LE as shown in Lemma \ref{lemmax2}, which is one of the disadvantages of FD-LE compared to FD-DFE.  

{\it Remark 6:}  We note that the results in Corollary \ref{lemma3} are obtained by assuming that there is no error propagation. Furthermore, we note that not all $NM$ data streams can benefit from the full diversity gain.  The simulation results provided in Section \ref{section 7} show that the  diversity orders  achievable for $x_{0}[k,l]$, $k<N-1$ and $l<M-1$, are smaller than that for $x_{0}[N-1,M-1]$.

\section{Downlink OTFS-NOMA  - Detecting   the NOMA Users' Signals}\label{section 5} 
Successive interference cancellation (SIC) will be carried out by the NOMA users, where each NOMA user first decodes the high mobility user's signal in the delay-Doppler plane and then decodes its own signal in the time-frequency plane. The two stages of SIC are discussed in the following two subsections, respectively. 

\vspace{-1em}

\subsection{Stage I of SIC}
Following   steps similar to the ones in the previous section, each NOMA user also observes the mixture of the $(M+1)$ users' signals in the delay-Doppler plane as follows:
\begin{align}\label{yx modelz1}
 \mathbf{y}_i =\gamma_0\mathbf{H}_i \mathbf{x}_0 +\underset{\text{Interference and noise terms}}{\underbrace{\sum^{M}_{q=1}\gamma_q\mathbf{H}_i {\mathbf{x}}_q +\mathbf{z}_i}},
\end{align}
where $\mathbf{H}_i$ and $\mathbf{z}_i$ are defined similar to $\mathbf{H}_0$ and $\mathbf{z}_0$, respectively. 

We assume  that the  low-mobility  NOMA  users  do not experience Doppler shift, and therefore, their channels can be simplified as follows:
\begin{align}\label{dd channel2}
 h_i(\tau) = \sum^{P_i}_{p=1}h_{i,p}\delta(\tau-\tau_{i,p}) ,
\end{align}
for $1\leq i\leq K$, 
which  means that each NOMA user's channel matrix, $\mathbf{H}_i$, $1\leq i\leq N$, is a block-diagonal matrix,  
i.e., $\mathbf{A}_{i,0}$ is a non-zero circulant matrix and $\mathbf{A}_{i,n}=\mathbf{0}_{M\times M}$, for $1\leq n \leq N-1$.  Therefore, each NOMA user can divide its observation vector into $N$ equal-length sub-vectors, i.e., $\mathbf{y}_i=\begin{bmatrix} \mathbf{y}_{i,0}^T &\cdots &\mathbf{y}_{i,N-1}^T \end{bmatrix}^T$, which yields the following simplified system model:
\begin{align}\label{yx modelz2}
 \mathbf{y}_{i,n} =\gamma_0 {\mathbf{A}}_{i,0} \mathbf{x}_{0,n} +\sum^{M}_{q=1}\gamma_q\mathbf{A}_{i,0} {\mathbf{x}}_{q,n} +\mathbf{z}_{i,n},
\end{align}
where, similar to  $\mathbf{y}_{i,n}$, $\mathbf{x}_{i,n}$ and $\mathbf{z}_{i,n}$ are obtained from $\mathbf{x}_{i}$ and $\mathbf{z}_{i}$, respectively.  Therefore, unlike the high-mobility user, the NOMA users can perform  their signal detection   based on   reduced-size observation vectors, which reduces the computational complexity. 

Since $\mathbf{A}_{i,0}$ is a circulant matrix, the two equalization approaches used in the previous section are still applicable.  First, we consider the use of FD-LE.  
Following  the same steps as in the proof for  Proposition \ref{proposition1}, in the first step of FD-LE,  the FFT matrix is applied  to the reduced-size observation vector, which yields  the following: 
\begin{align}\label{yx model6z}
\tilde{\mathbf{y}}_{i,n}
 = & \tilde{\mathbf{D}}_i    \mathbf{F}_M^H\left(\gamma_0\mathbf{x}_{0,n} +\sum^{M}_{q=1}\gamma_q  {\mathbf{x}}_{q,n} \right) +\tilde{\mathbf{z}}_{i,n},
\end{align} 
where $\tilde{\mathbf{y}}_{i,n}=  \mathbf{F}_M^H   \mathbf{y}_{i,n} $ and $\tilde{\mathbf{z}}_{i,n}= \mathbf{F}_M^H \mathbf{z}_{i,n}$. Compared to  $\mathbf{D}_i$ in Proposition \ref{proposition1} which is an $NM\times NM$ matrix, $\tilde{\mathbf{D}}_i$ is an $M\times M$ diagonal matrix, and its $(l+1)$-th diagonal element    is given by  $
\tilde{D}_{i}^{l}=  \sum^{M-1}_{m=0} a_{i,0}^{m,1} e^{j2\pi \frac{lm}{M}}  $,
for  $0\leq l \leq M-1$, where 
 $a^{m,1}_{i,0}$ is the element located in the $(m+1)$-th row and the first column of $\mathbf{A}_{i,0}$. Unlike conventional OFDM, which uses $\mathbf{F}_M$ at the receiver, $\mathbf{F}_M^H$ is used here. Because $ \mathbf{F}_M^H\mathbf{A}_{i,0}\mathbf{F}_M = \left[\mathbf{F}_M\mathbf{A}^*_{i,0}\mathbf{F}_M ^H\right]^*$,  the sign of the exponent of the exponential component of $\tilde{D}_{i}^{l}$ is different from that in the conventional case.
 
In the second step of FD-LE,  $\mathbf{F}_M\tilde{\mathbf{D}}_i ^{-1}$ is applied to $\tilde{\mathbf{y}}_{i,n}$. Following steps similar to the ones in the proof for Lemma \ref{lemma1}, the SINR for detecting $x_{0}[k,l]$ can be obtained as follows: 
 \begin{align}
\text{SINR}_{0,kl}^{i, \text{LE}} = \frac{\rho\gamma_0^2}{\rho \gamma^2_1 +\frac{1}{M}\sum^{M-1}_{\tilde{l}=0} |\tilde{D}_i^{\tilde{l}}|^{-2}}.
\end{align}
We note that  $\text{SINR}_{0,k_1l}^{i, \text{LE}}=\text{SINR}_{0,k_2l}^{i, \text{LE}}$, for $k_1\neq k_2$,   due to the time invariant nature of the channels. 
 
 If FD-DFE is used,   the corresponding SINR for detecting  $x_0[k,l]$ is given by
 \begin{align}\label{SINR approach IIx}
\text{SINR}_{0,kl}^{i, \text{DFE}} = \frac{\rho \gamma_0^2}{\rho\gamma_1^2 + \tilde{\lambda}_{0,l}^{-1}},
\end{align}
where $\tilde{\lambda}_{0,l}$ is obtained from the Cholesky decomposition of $\mathbf{A}_{i,0}$. The details for the derivation are omitted here due to space limitations. 

\vspace{-1em}

\subsection{Stage II of SIC} \label{subsection 2}
Assume that $\text{U}_0$'s $NM$ signals can be decoded and  removed  successfully, which means that, in the time-frequency plane, the NOMA users observe the following: 
\begin{align} 
Y_i[n,m]   =&  \sum^{M}_{q=1}\gamma_q H_i[n,m] X_q[n,m] +W_{i}[n,m]
=  \gamma_1 H_i[n,m] x_{m+1}(n) +W_i[n,m]
,
\end{align}
where the last step follows from the mapping scheme used in \eqref{indooruser2} and   it is  assumed that all   NOMA users employ  the same power allocation coefficient. 
We note that  $\text{U}_{i}$ is only interested in   $Y_{i}[n, i-1]$, $0\leq n \leq N-1$.  
Therefore, $\text{U}_i$'s $n$-th information bearing signal, $x_i(n)$, can be detected by applying a one-tap equalizer as follows:
\begin{align}
\hat{x}_i(n) = & \frac{Y_i[n,i-1]}{\gamma_1H_i[n,i-1]}  ,
\end{align}
which means that the SNR for detecting $x_i(n)$ is given by
\begin{align}
\text{SNR}_{i,n} = \rho \gamma_1^2 |\tilde{D}_i^{i-1}|^2,
\end{align}
since $W_i[n,i-1]$ is white Gaussian noise and $H_i[n,i-1] = \tilde{D}_i^{i-1}$. We note that $\text{SNR}_{i,n_1}=\text{SNR}_{i,n_2}$, for $n_1\neq n_2$, which is due to the time-invariant  nature of the channel. 
 
 Without loss of generality, assume that the same target data rate $R_i$ is used for $x_{i}(n)$, $0\leq n \leq N-1$.  
Therefore, the outage probability for $x_i(n)$  is given by
\begin{align}
\mathrm{P}_{i,n}^{\text{LE}} = &1 -  \mathrm{P}\left( \text{SNR}_{i,n} >\epsilon_i, \text{SINR}_{0,kl}^{i,\text{LE}} >\epsilon_0, \forall l\right)\\\nonumber
= &1 -  \mathrm{P}\left( \rho \gamma_1^2 |\tilde{D}_i^{i-1}|^2 >\epsilon_i, \frac{\rho\gamma_0^2}{\rho \gamma^2_1 +\frac{1}{M}\sum^{M-1}_{l=0} |\tilde{D}_i^{l}|^{-2}}>\epsilon_0 \right),
\end{align}
if FD-LE is used in the first stage of SIC. If FD-DFE is used in the first stage of SIC, the outage probability for $x_i(n)$  is given by
\begin{align}
\mathrm{P}_{i,n}^{\text{DFE}} = &1 -  \mathrm{P}\left( \text{SNR}_{i,n} >\epsilon_i, \text{SINR}_{0,kl}^{i,\text{DFE}} >\epsilon_0, \forall l\right)\\\nonumber= & 1 -  \mathrm{P}\left( \rho \gamma_1^2 |\tilde{D}_i^{i-1}|^2 >\epsilon_i, \frac{\rho\gamma_0^2}{\rho \gamma^2_1 + \tilde{\lambda}_{0,l}^{-1}}>\epsilon_0, \forall l\right),
\end{align}
where $\epsilon_i = 2^{R_i}-1$. 
 Again because of the correlation  between the random variables $|\tilde{D}_i^{l}|^{-2}$ and $\tilde{\lambda}_{0,l}$, the exact expressions for the outage probabilities are difficult to obtain. Alternatively, the achievable diversity order   is analyzed in the following subsections. 

\subsubsection{Random User Scheduling} If the $M$ users are randomly selected from the $K$ available  users, which means that each $ |\tilde{D}_i^l|^2$ is  complex Gaussian distributed. For the FD-LE case, the outage probability, $\mathrm{P}_{i,n} ^{\text{LE}}  $, can be upper bounded as follows:
\begin{align}\label{eq371}
&
\mathrm{P}_{i,n} ^{\text{LE}} 
\leq  1 -  \mathrm{P}\left( \rho \gamma_1^2 |\tilde{D}^{\min}_i|^2 >\epsilon_i, \frac{\rho\gamma_0^2}{\rho \gamma^2_1 +  |\tilde{D}_i^{\min}|^{-2}}>\epsilon_0\right),
\end{align}
where $|\tilde{D}_i^{\min}|^{2}= \min \{  | \tilde{D}_{i}^{m}|^2, 0\leq m \leq M-1  \}$.
The upper bound on the outage probability in \eqref{eq371} can be rewritten  as follows:
\begin{align}
 \mathrm{P}_{i,n}  ^{\text{LE}}
\leq  1 -  \mathrm{P}\left(  |\tilde{D}^{\min}_i|^2 >\bar{\epsilon}\right),
\end{align}
where $\bar{\epsilon} = \max\left\{ \frac{\epsilon_0}{\rho( \gamma_0^2-\gamma_1^2\epsilon_0)}, \frac{\epsilon_i}{\rho\gamma_1^2}\right\}$. As a result,  an upper bound on the outage probability can be obtained as follows:
\begin{align}
 \mathrm{P}_{i,n}  ^{\text{LE}}\leq \mathrm{P}\left(  |\tilde{D}^{\min}_i|^2 <\bar{\epsilon}\right)\leq M \mathrm{P}\left( | \tilde{D}^{0}_i|^2 <\bar{\epsilon}\right)\doteq \frac{1}{\rho},
\end{align}
where $\mathrm{P}^o\doteq \rho^{-d}$ denotes   exponential equality, i.e., $d= - \underset{\rho\rightarrow \infty}{\lim}\frac{\log\mathrm{P}^o}{\log \rho}$ \cite{Zhengl03}. Therefore, the following corollary can be obtained.

\begin{corollary}\label{corollary3}
A diversity order of $1$ is achievable at the NOMA users for  the FD-LE approach. 
\end{corollary}
Our simulation results in Section \ref{section 7} show that a diversity order of $1$ is also achievable for FD-DFE, although we do not have a formal proof for this conclusion, yet.

\subsubsection{Realizing  Multi-User Diversity}
 The diversity order of OTFS-NOMA can be improved by carrying out opportunistic user scheduling, which yields  multi-user diversity gains.    For illustration purpose, we propose a greedy user scheduling policy, where a single NOMA user   is scheduled to transmit in all resource blocks of the time-frequency plane. From the analysis of the random   scheduling case we deduce  that $|\tilde{D}_i^{\text{min}}|^2$ is critical to the outage performance. Therefore, the scheduled NOMA user,   denoted by $\text{U}_{i^*}$, is selected based on the following criterion: 
\begin{align}\label{im selecy}
i^* = \arg \underset{i \in\{1, \cdots, K\}}{\max} \left\{|\tilde{D}_i^{\min}|^{2} \right\}.
\end{align}

By using the assumption  that the users' channel gains are independent and following steps similar to the ones in the proof for Lemma \ref{lemmax2}, the following corollary can be obtained in a straightforward manner.

\begin{corollary}\label{corollary3}
For FD-LE, the user scheduling strategy shown in \eqref{im selecy} realizes the maximal multi-user diversity gain, $K$. 
\end{corollary} 
We note that the user scheduling strategy shown in \eqref{im selecy} is also useful for improving  the performance of  FD-DFE, as shown in Section \ref{section 7}.

\section{Uplink OTFS-NOMA Transmission }  \label{section 6}
The design of uplink OTFS-NOMA is similar to that of downlink OTFS-NOMA, and due to space limitations, we mainly focus on the difference between the two cases in this section.  Again consider that $\text{U}_0$ is grouped with $M$ NOMA users, selected from the $K$ available  users.  $\text{U}_0$'s $NM$ signals are placed in the delay-Doppler plane, denoted by $x_0[k,l]$, where $0\leq k \leq N-1$ and $0\leq l \leq M-1$. The corresponding  time-frequency signals, $X_0[n,m]$, are obtained by applying ISFFT to $x_0[k,l]$. On the other hand, the NOMA users' signals, $x_i(n)$, are mapped to   time-frequency signals, $X_i[n,m]$, according to \eqref{indooruser2}.

Following  steps similar to the ones for the downlink case, the    base station's  observations in the time-frequency plane are given by 
\begin{align}\label{TF1}
Y[n,m] =& \sum^{M}_{q=0}H_q[n,m]X_q[n,m] +W[n,m]\\\nonumber
=&\frac{H_0(n,m)}{NM}\sum^{N-1}_{k=0}\sum^{M-1}_{l=0}x_0[k,l]e^{j2\pi \left(\frac{kn}{N}-\frac{ml}{M}\right)} +\sum^{M}_{q=1}  H_q[n,m] X_q[n,m]+W[n,m],
\end{align}
where $W[n,m]$ is the Gaussian noise at the base station in the time-frequency plane. We assume that   all  users employ  the same transmit pulse as well as the same transmit power.   The base station  applies SIC to first detect the NOMA users' signals in the time-frequency plane, and then tries  to detect the high-mobility user's signals in the delay-Doppler plane, as shown in the following two subsections.  

\vspace{-1em}

\subsection{ Stage I of SIC}
The base station will first try to detect the NOMA users' signals in the time-frequency plane by treating the signals from $\text{U}_0$ as noise, which is the first stage of SIC. 

By using \eqref{indooruser2}, $ {x}_i(n) $ can be estimated as follows: 
 \begin{align}\label{TF2}
\hat{x}_i(n) =& \frac{Y[n, i-1]}{H_i[n,i-1]}  =  x_i[n]+ \frac{H_0[n,i-1]X_0[n,i-1] +W[n,i-1]}{H_i[n,i-1]}. 
\end{align}
Define an $NM\times 1$ vector, $\bar{\mathbf{x}}_0$, whose $(nM+m+1)$-th element is $X_0[n,m]$. Recall that  $X_0[n,m]$ is obtained from the ISFFT of $x_0[k,l]$, i.e.,
\begin{align}
\bar{\mathbf{x}}_0 =& (\mathbf{F}_N^H\otimes  \mathbf{F}_M) \mathbf{x}_0,
\end{align}
which means $X_0[n,m]$ follows the same distribution as $x_0[k,l]$.    By applying  steps similar to those in the proof for Lemma \ref{lemma1},  the SINR for detecting  $x_i(n)$ is given by
\begin{align}
\text{SINR}_{i,n} = \frac{\rho|H_i[n,i-1]|^2}{\rho |H_0[n,i-1]|^2 +1}.
\end{align}

Unlike downlink OTFS-NOMA, there are two possible strategies for uplink OTFS-NOMA to combat   multiple access interference, as shown in the following two subsections. 

\subsubsection{Adaptive-Rate Transmission}
One strategy to combat multiple access interference is to  impose the following constraint on $x_i(n)$:
\begin{align}\label{ri adaptive}
R_{i,n}\leq  \log\left(1+ \frac{\rho|H_i[n,i-1]|^2}{\rho |H_0[n,i-1]|^2 +1}\right),
\end{align}
which means that the first stage of SIC is guaranteed to  be successful. Therefore,  the impact of the NOMA users on $\text{U}_0$'s performance is minimized, i.e.,   the use of NOMA is transparent to $\text{U}_0$. 


Because $\text{U}_i$'s data rate is adaptive, outage events when decoding  $x_i(n)$  do not happen, which means that an appropriate criterion for the performance evaluation is the ergodic rate. Recall that $H_i[n,i-1] = \tilde{D}_i^{i-1}$ and $H_0[n,i-1]=D_0^{n,i-1}$. Therefore, $\text{U}_i$'s ergodic rate is given by
\begin{align} \label{ergodic x}
\mathcal{E}\{R_{i,n}\}\leq \mathcal{E}\left\{ \log\left(1+ \frac{\rho| \tilde{D}_i^{i-1}|^2}{\rho |D_0^{n,i-1}|^2 +1}\right)\right\}.
\end{align}

We note that the ergodic rate of uplink OTFS-NOMA can be further  improved by modifying  the user scheduling strategy proposed in   \eqref{im selecy}, as shown in the following. Particularly, denote 
  the NOMA user which is scheduled to transmit in the $m$-th frequency subchannel  by $\text{U}_{i_m^*}$, and this user  is selected by using the following criterion: 
\begin{align}\label{im selecy2}
i_m^* = \arg \underset{i\in\{1, \cdots, K\}}{\max} \left\{   | \tilde{D}_{i}^{m}|^2\right\}.
\end{align}
It is worth pointing out that a single user might be scheduled on multiple frequency channels, which reduces user fairness.  

Because   the integration of the logarithm function appearing  in \eqref{ergodic x} leads to    non-insightful special functions, we will use simulations   to evaluate the ergodic rate of OTFS-NOMA in Section \ref{section 7}. 
\subsubsection{Fixed-Rate Transmission} 
If the NOMA users do not have the capabilities to adapt their transmission rates, they have to use   fixed  data rates $R_i$ for transmission, which means that outage events can happen and the achieved outage performance is analyzed in the following.   For illustration purposes, we focus on the case when the user scheduling strategy shown in  \eqref{im selecy2} is used. 

The outage probability for  detecting  $x_{i_m^*}(n)$ is given by 
\begin{align}
 \mathrm{P}_{i_m^*,n} =\mathrm{P}\left( \log\left(1+\frac{\rho|\tilde{D}_{i^*_m}^{i^*_m-1}|^2}{\rho |D_0^{n,i^*_m-1}|^2 +1}\right)<R_{i_m^*}\right).
 \end{align}
 Following   steps similar to the ones  in the proof for Lemma \ref{lemmax2}, we can show  that  $|\tilde{D}_{i^*_m}^{i^*_m-1}|^2$ and $|D_0^{n,i^*_m-1}|^2$ are independent,  and the use of the user scheduling scheme in \eqref{im selecy2} simplifies  the outage probability   as follows:
\begin{align}
 \mathrm{P}_{i_m^*,n} =&\mathrm{P}\left( \log\left(1+\frac{\rho|\tilde{D}_{i^*_m}^{i^*_m-1}|^2}{\rho |D_0^{n,i^*_m-1}|^2 +1}\right)<R_{i_m^*}\right)
 =
\int^{\infty}_{0} \left(1-e^{-\frac{\epsilon_{i_m^*}(1+\rho y)}{\rho}}\right)^{K} e^{-y} dy,
 \end{align}
 where we use the fact that  the cumulative distribution  function of  $|\tilde{D}_{i^*_m}^{i^*_m-1}|^2$ is $  \left(1-e^{-x}\right)^{K}$ because of the adopted  user scheduling strategy. 

The outage probability can be further simplified as follows:
\begin{align}
 \mathrm{P}_{i_m^*,n}  =& 
\sum^{K}_{k=0}{K\choose k} (-1)^k \int^{\infty}_{0} e^{-\frac{k\epsilon_{i_m^*}(1+\rho y)}{\rho}-y}    dy 
 =
\sum^{K}_{k=0}{K\choose k} (-1)^k e^{-\frac{k\epsilon_{i_m^*}}{\rho}}   \frac{1}{k\epsilon_{i_m^*} +1}.     
 \end{align}
 
 At high SNR, the outage probability can be approximated as follows:  
\begin{align}\label{highx}
 \mathrm{P}_{i_m^*,n}  \approx& 
\sum^{K}_{k=0}{K\choose k} (-1)^k    \frac{1}{k\epsilon_{i_m^*} +1}    ,
 \end{align}
 which is no longer a function of $\rho$, i.e.,  the outage probability has an error floor at high SNR. This is due to the fact that $\text{U}_{i^*_m}$ is subject to strong interference from $\text{U}_0$.  
 
 However, we can show that  the error floor experienced by $\text{U}_{i^*_m}$ can be reduced by increasing $K$, i.e., inviting more opportunistic users for NOMA transmission. In particular, assuming  $K\epsilon_{i_m^*}\rightarrow 0$, the outage probability can be approximated as follows: 
 \begin{align}
 \mathrm{P}_{i_m^*,n}  \approx& 
\sum^{K}_{k=0}{K\choose k} (-1)^k   \left(1+k\epsilon_{i_m^*}\right)^{-1}
 \approx
\sum^{K}_{k=0}{K\choose k} (-1)^k   \sum^{\infty}_{l=0} (-1)^l k^l\epsilon_{i_m^*}^l,
 \end{align}
 where we use the fact that $(1+x)^{-1}= \sum^{\infty}_{l=0} (-1)^lx^l$, $|x|<1$. Therefore, the error floor at high SNR can be approximated as follows:
  \begin{align}
 \mathrm{P}_{i_m^*,n}    \approx& 
\sum^{\infty}_{l=0} (-1)^l \epsilon_{i_m^*}^l\sum^{K}_{k=0}{K\choose k} (-1)^k   k^l 
\approx
  (-1)^K \epsilon_{i_m^*}^K(-1)^KK! = K! \epsilon_{i_m^*}^K,
 \end{align}
 where we use the identities     
 $
 \sum^{K}_{k=0}{K\choose k} (-1)^k   k^l = 0$, 
 for $l<K$ and $
 \sum^{K}_{k=0}{K\choose k} (-1)^k   k^K = (-1)^KK! $.

The conclusion that increasing $K$   reduces the error floor can be confirmed     by defining  $f(k) = k! \epsilon_{i_m^*}^k$ and using the following fact:
 \begin{align}\label{deiff}
 f(k) - f(k+1) =  k! \epsilon_{i_m^*}^k\left(1- (k+1) \epsilon_{i_m^*}\right)>0,
 \end{align}
 where it is assumed that $k\epsilon_{i_m^*}\rightarrow 0$.

 \vspace{-1em}
 
\subsection{Stage II of SIC}
If adaptive transmission is used, the NOMA users' signals can be detected successfully during the first stage of SIC. Therefore,  they can be removed from the observations at the base station, i.e.,  $\bar{Y} [n,m]= Y[n,m] - \sum^{N}_{q=1} H_q(n,m) X_q[n,m]$, and   SFFT is applied to obtain the delay-Doppler observations  as follows:
\begin{align}
y_0[k,l] =&\frac{1}{NM}\sum^{N-1}_{n=0}\sum^{M-1}_{m=0} \bar{Y}[n,m] e^{-j2\pi \left(\frac{nk}{N}-\frac{ml}{M}\right)} 
=\sum^{P_0}_{p=1} h_{0,p} x_0[(k-k_{\mu_{0,p}})_N, (l-l_{\tau_{0,p}})_M]   +z[k,l], 
\end{align}
where $z[k,l]$ denote additive noise.  $\text{U}_0$'s signals can be detected by applying either of the    two considered equalization approaches,  and the same performance as for  OTFS-OMA can be realized.  The analytical development is similar to the downlink case, and hence is omitted due to space limitations. 

However, if fixed-rate transmission  is used, the uplink outage events for decoding $x_0[k,l]$  are different from the downlink ones, as shown in the following. Particularly, the use of FD-LE yields the following 	SINR expression  for decoding $x_0[k,l]$:
\begin{align}
\text{SINR}_{0,kl}^{\text{LE}} = \frac{\rho }{ \frac{1}{NM}\sum^{N-1}_{k=0}\sum^{M-1}_{l=0} |D_0^{k,l}|^{-2}}.
\end{align}
If FD-DFE is used, the SNR for detection of $x_0[k,l]$ is given by
\begin{align} 
\text{SINR}_{0,kl}^{\text{DFE}} =  \rho   \lambda_{0,kl}.
\end{align}
Therefore, the outage probability for detecting $x_0[k,l]$ is given by 
\begin{align}\nonumber
\mathrm{P}_{kl} = &1 -  \mathrm{P}\left(  \text{SINR}^{\text{DFE/LE}}_{0,kl} >\epsilon_0, \text{SNR}_{i,n}  >\epsilon_i \forall i,n\right)\\\nonumber
\geq  &   1 -  \mathrm{P}\left(   \text{SNR}_{i,n} >\epsilon_i \forall i,n\right)\geq    \mathrm{P}\left(   \text{SNR}_{1,0} <\epsilon_i \right).
\end{align} 
Since $ \mathrm{P}\left(   \text{SNR}_{1,0} <\epsilon_i \right)$ has an error floor as shown in the previous subsection,     the uplink outage probability for detection of  $\text{U}_0$'s signals  does not go to zero even if $\rho\rightarrow \infty$, which is different from the downlink case.
\vspace{-1em}

\section{Numerical Studies} \label{section 7}
\begin{table}[t]
\caption{Delay-Doppler Profile for $\text{U}_0$'s channel} 
\centering 
\vspace{-1em}
\begin{tabular}{l c c c c } 
\hline\hline 
Propagation path index ($p$) & 0& 1 & 2 &3  \\ [0.5ex] 
\hline 
Delay ($\tau_{0,p}$) $\mu s$ & 8.33 & 25 & 41.67 &58.33 \\ 
Delay tap index ($l_{\tau_{0,p}}$ ) & 2 & 6 & 10 &14  \\
Doppler ($\nu_{0,p}$) Hz  & 0 & 0 & 468.8  &468.8   \\
Doppler tap index ($k_{\nu_{0,p}}$  ) & 0 & 0& 1  &1\\
\hline 
\end{tabular}
\label{table:nonlin} 
\vspace{-2em}
\end{table}

In this section, the performance of OTFS-NOMA is evaluated via computer simulations. Similar to \cite{8503182, 8516353,8424569}, we first define   the delay-Doppler profile for $\text{U}_0$'s channel as shown in  Table \ref{table:nonlin}, where   $P_0=3$ and the subchannel spacing is $\Delta f = 7.5$ kHz.  Therefore, the  maximal speed corresponding to the largest Doppler shift $\nu_{0,3}=468.8 $ Hz is $126.6$ km/h if the carrier frequency is $f_c=4$ GHz.  On the other hand,  the NOMA users' channels are assumed to be time invariant with  $P_i=3$ propagation paths, i.e., $\tau_{i,p}=0$ for $p\geq 4$, $i\geq 1$.  For all the users' channels, we assume that $\sum^{P_i}_{p=0}\mathcal{E}\{|h_{i,p}|^2\}=1$ and  $ |h_{i,p}|^2 \sim CN\left(0,\frac{1}{P_i+1}\right)$.

\begin{figure}[!htp] \vspace{-1em}
\begin{center}\subfigure[$R_0=1$ BPCU and $R_i=1.5$ BPCU]{\label{fig1b}\includegraphics[width=0.43\textwidth]{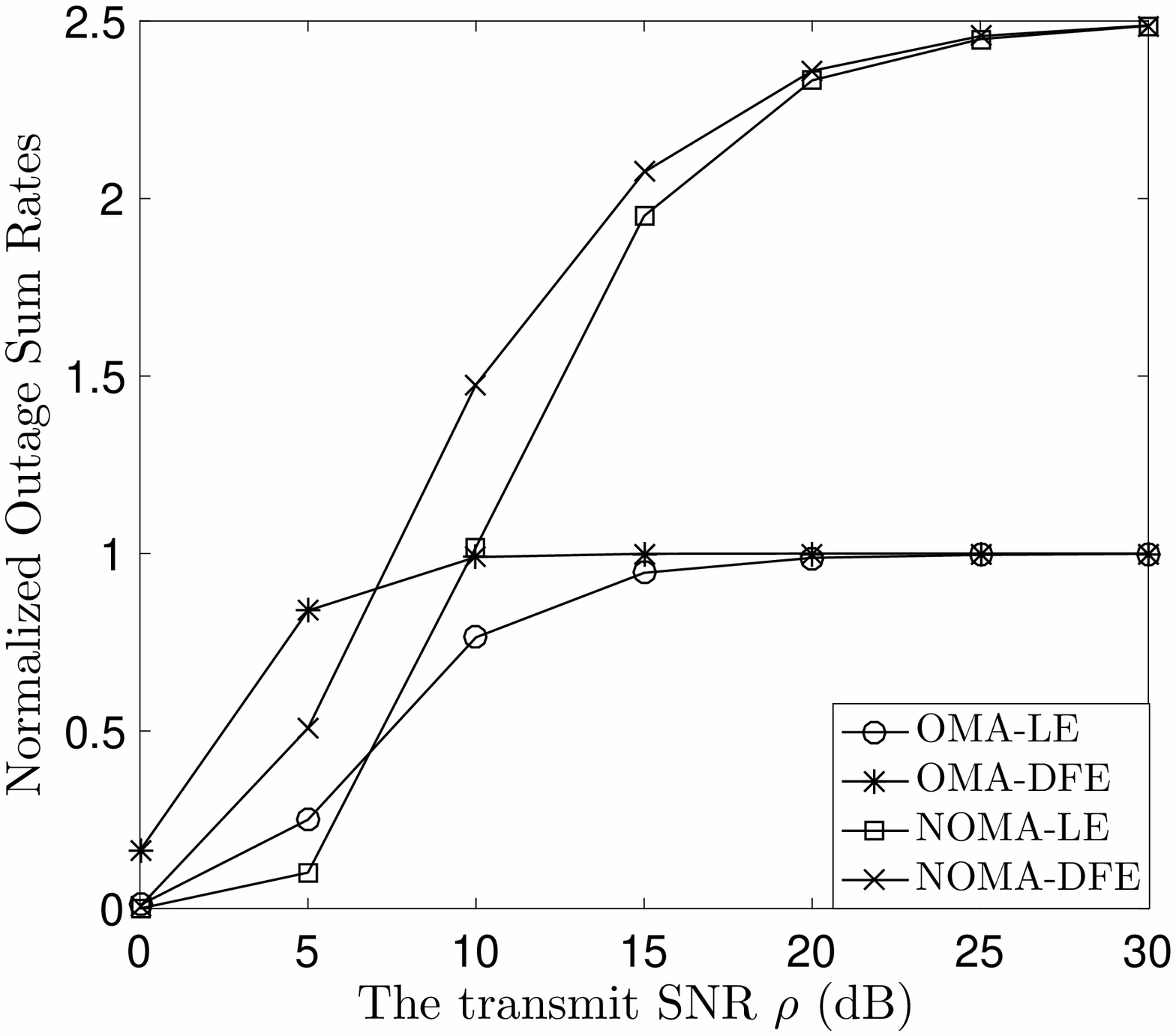}}
\subfigure[$R_0=0.5$ BPCU and $R_i=1$ BPCU]{\label{fig1a}\includegraphics[width=0.43\textwidth]{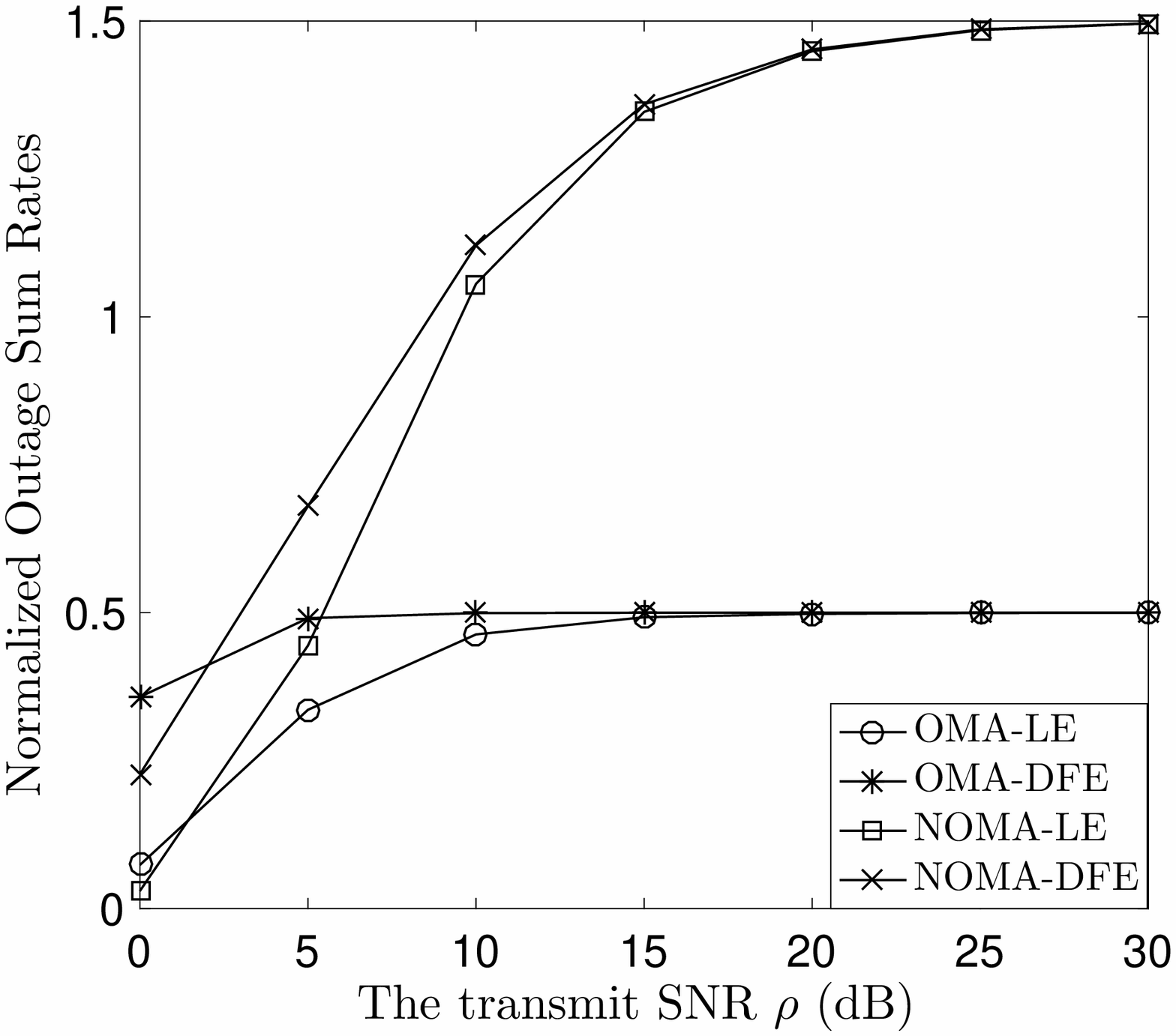}} 
\end{center}\vspace{-1em}
 \caption{Impact of OTFS-NOMA on the downlink sum rates. $M=N=K=16$. $P_0=P_i=3$. BPCU denotes bit per channel use.   $\gamma_0^2=\frac{3}{4}$ and $\gamma_i^2=\frac{1}{4}$ for $i>0$.  Random user scheduling is used.    }\label{fig 1}\vspace{-2em}
\end{figure}

In Fig. \ref{fig 1},  downlink OTFS-NOMA transmission is evaluated   by using the normalized outage sum rate  as the performance criterion   which is  defined as $\frac{1}{NM}\sum^{N-1}_{k=0}\sum^{M-1}_{l=0}(1-\mathrm{P}_{0,kl})R_0$ and $\frac{1}{NM}\sum^{N-1}_{k=0}\sum^{M-1}_{l=0}(1-\mathrm{P}_{0,kl})R_0+\frac{1}{NM}\sum^{M}_{i=1}\sum^{N-1}_{n=0}(1-\mathrm{P}_{i,n})R_i$ for OTFS-OMA and OTFS-NOMA, respectively.   Fig.~\ref{fig 1} shows that the use of OTFS-NOMA can significantly improve the sum rate at high SNR for both considered  choices of $R_0$ and $R_i$. The reason for this performance gain is   the fact that    the maximal sum rate achieved by OTFS-OMA is capped by   $R_0$, whereas  OTFS-NOMA can provide sum rates up to  $R_0+R_i$.  Comparing Fig. \ref{fig1a} to Fig. \ref{fig1b}, one can observe  that  the performance loss of OTFS-NOMA at low SNR can be mitigated by reducing the target data rates, since reducing the target rates improves the   probability of successful SIC. Furthermore, both   figures  show that FD-DFE outperforms FD-LE in the entire considered range of SNRs; however, we note that the performance gain of FD-DFE over FD-LE is achieved at the expense  of increased computational complexity. 

\begin{figure}[!htp] \vspace{-1em}
\begin{center}\subfigure[Outage probabilities of $\text{U}_0$ and the NOMA users]{\label{fig2a}\includegraphics[width=0.43\textwidth]{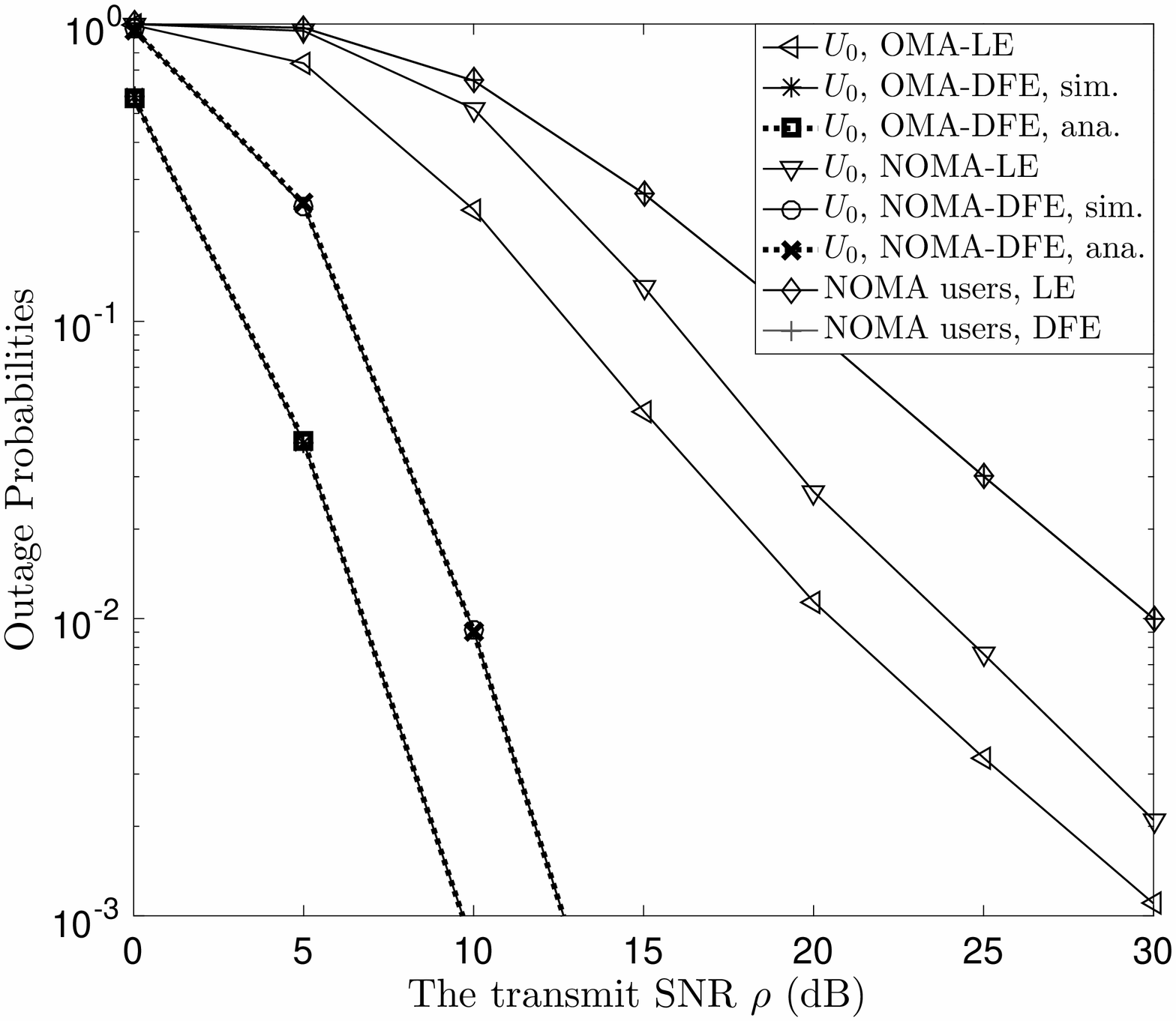}} \subfigure[Performance  of FD-DFE]{\label{fig2b}\includegraphics[width=0.43\textwidth]{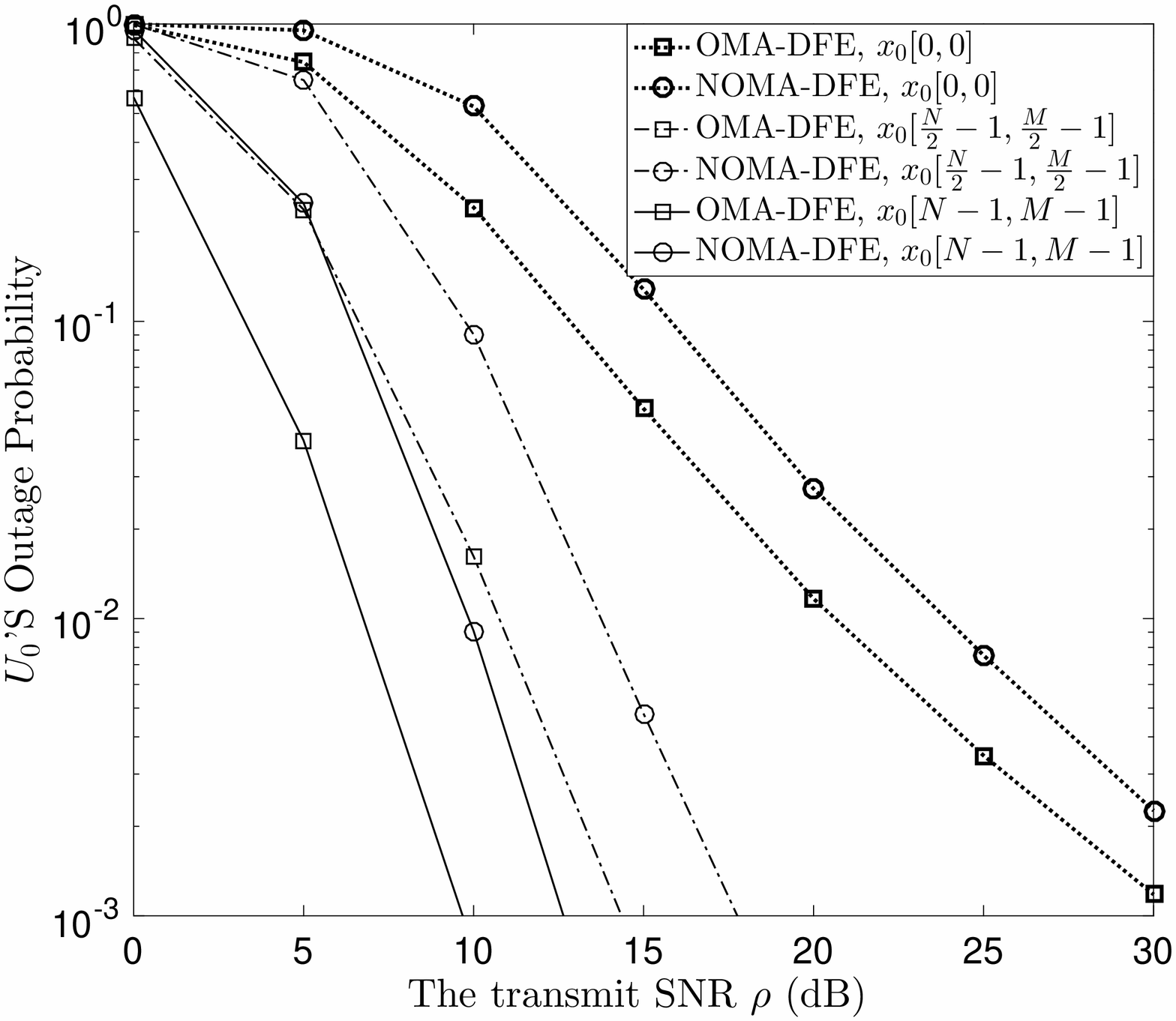}}
\end{center}\vspace{-1em}
 \caption{The outage performance of downlink OTFS-OMA and OTFS-NOMA.  $M=N=K=16$. $P_0=P_i=3$.     $\gamma_0^2=\frac{3}{4}$ and $\gamma_i^2=\frac{1}{4}$ for $i>0$.  $R_0=0.5$ BPCU and $R_i=1$ BPCU. In Fig. \ref{fig2a}, for   FD-DFE, the performance of $x_0[N-1,M-1]$ is shown.   Random user scheduling is used.    }\label{fig2}\vspace{-2em}
\end{figure}
 
In Fig. \ref{fig2}, the outage probabilities achieved by downlink OTFS-OMA and OTFS-NOMA are shown. As can be seen from Fig. \ref{fig2a}, the diversity order achieved with FD-LE for detection of $x_0[k,l]$  is one, as expected from   Lemma \ref{lemmax2}. As discussed in Section \ref{subsection fd-dfe}, one advantage of FD-DFE over FD-LE is that  FD-DFE facilitates   multi-path fading diversity gains, whereas FD-LE is limited to a diversity gain of  one. This conclusion is confirmed  by Fig. \ref{fig2a}, where the analytical results developed in Corollary \ref{lemma3} are also verified.    Fig. \ref{fig2b} shows   the outage probabilities achieved by FD-DFE for different $x_0[k,l]$. As shown in the figure, the lowest outage probability is obtained for $x_0[N-1,M-1]$, whereas the outage probability  of $x_0[0,0]$ is the largest, which is due to the fact that, in FD-DFE,   different signals  $x_0[k,l]$ are affected by  different effective channel gains, $\lambda_{0,kl}$. Another important observation from the figures is that the FD-LE outage probability  is   the same as the FD-DFE outage probability for detection of $x_0[0,0]$, which fits the intuition that for FD-DFE the reliability  of  the first decision ($x_0[0,0]$) is the same as that of FD-LE.  For the same reason,    FD-LE and FD-DFE yield   similar performance for detection of   the NOMA users' signals, since the FD-DFE outage performance is dominated by the reliability for   detection of $x_0[0,0]$,  and hence is the same as that of FD-LE.  

In addition to multi-path diversity, another degree of freedom available  in the considered OTFS-NOMA downlink scenario is multi-user diversity, which can be harvested by applying user scheduling as discussed in Section \ref{subsection 2}.  Fig. \ref{fig4} demonstrates the benefits of exploiting  multi-user diversity. With random user scheduling, at low SNR, the performance of OTFS-NOMA is worse than that of OTFS-OMA, which is   also consistent with Fig. \ref{fig 1}. By increasing the number of   users participating in OTFS-NOMA, the performance of OTFS-NOMA can be   improved, particularly at low and moderate SNR. For example, for   FD-LE,   the performance of  OTFS-NOMA  approaches  that of OTFS-OMA at low SNR by exploiting   multi-user diversity, and for   FD-DFE, an extra gain of $0.5$ BPCU can be achieved at moderate SNR.

\begin{figure}[!htb]
    \centering
    \begin{minipage}{0.49\textwidth}
        \centering
   \epsfig{file=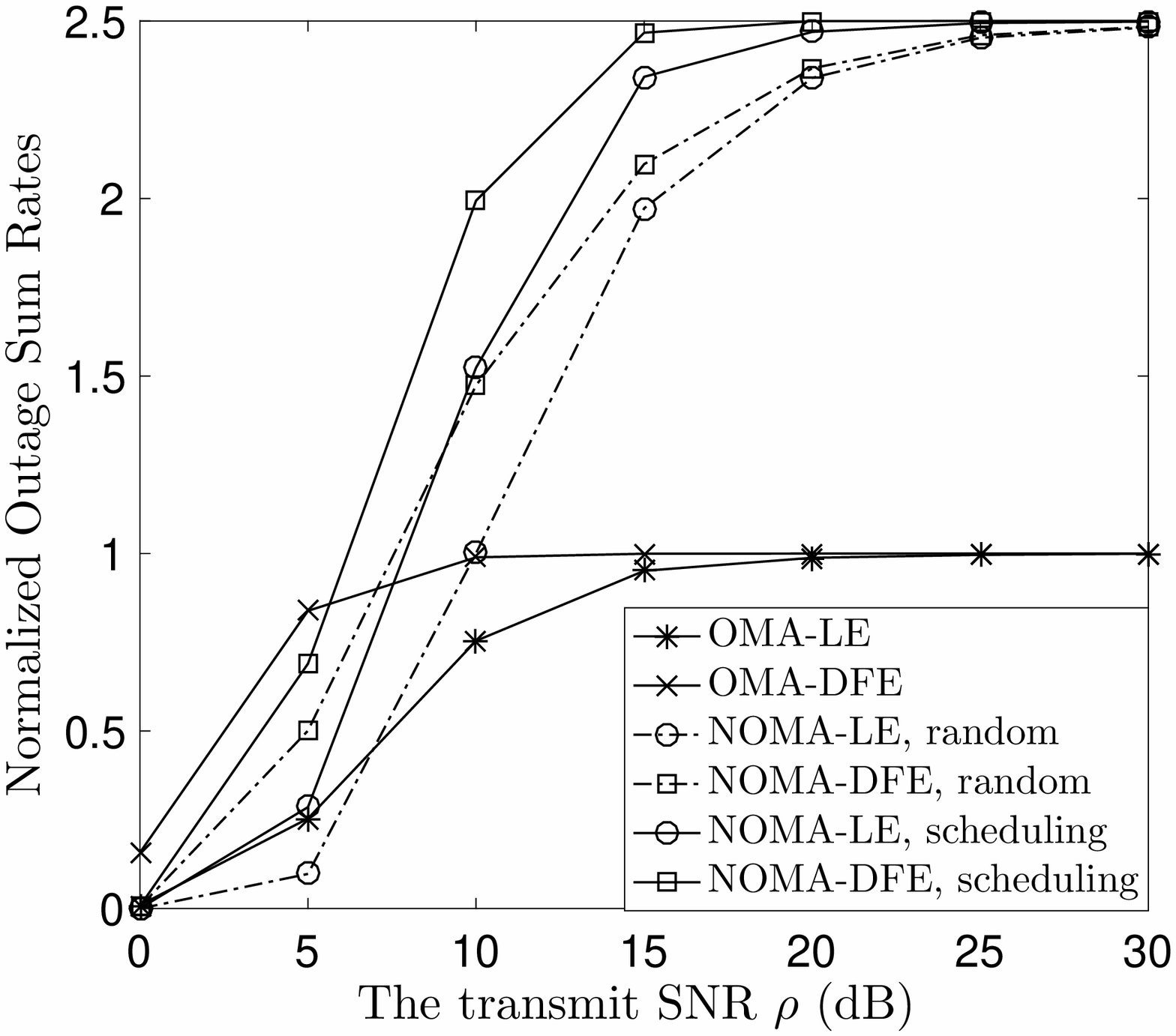, width=0.9\textwidth, clip=}\vspace{-1em}
\caption{ Impact of user scheduling   on the downlink outage sum rates. $P_0=P_i=3$. $R_0=1$  BPCU and $R_i=1.5$ BPCU. $M=N=K=16$, $\gamma_0^2=\frac{3}{4}$ and $\gamma_i^2=\frac{1}{4}$ for $i>0$.     }\label{fig4}
    \end{minipage}
    \hfill
    \begin{minipage}{0.49\textwidth}
        \centering
         \includegraphics[width=.9\textwidth]{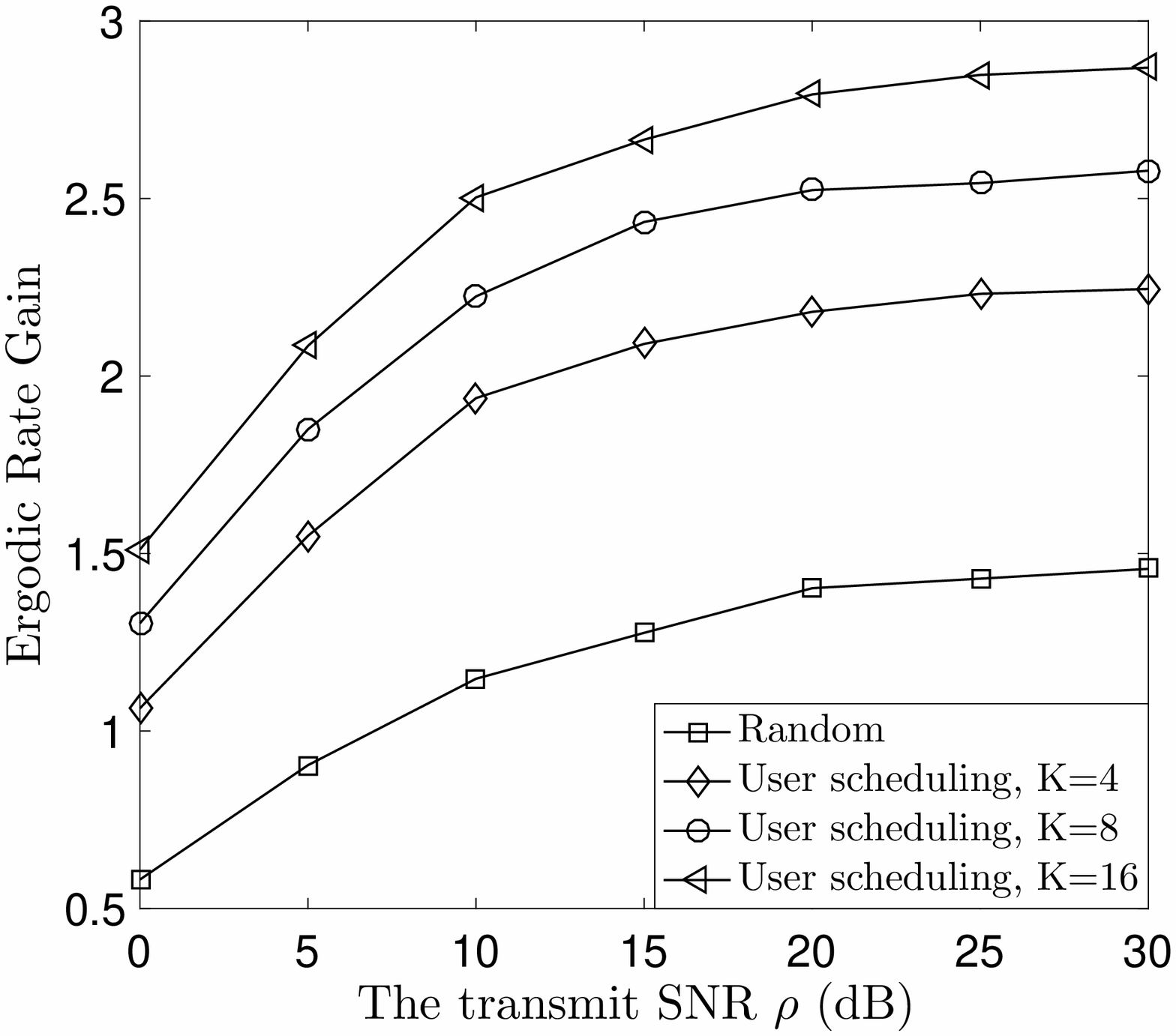} \vspace{-1em}
\caption{ The ergodic rate gain of OTFS-NOMA over OTFS-OMA.  The NOMA users adapt their data rates according to \eqref{ri adaptive}.  $P_0=P_i=3$.  $M=N=16$.    }\label{fig5}
    \end{minipage}\vspace{-1.5em} 
\end{figure}

\begin{figure}[!htp] 
\begin{center}\subfigure[Outage Sum Rate ($R_i=1$ BPCU)]{\label{fig6a}\includegraphics[width=0.43\textwidth]{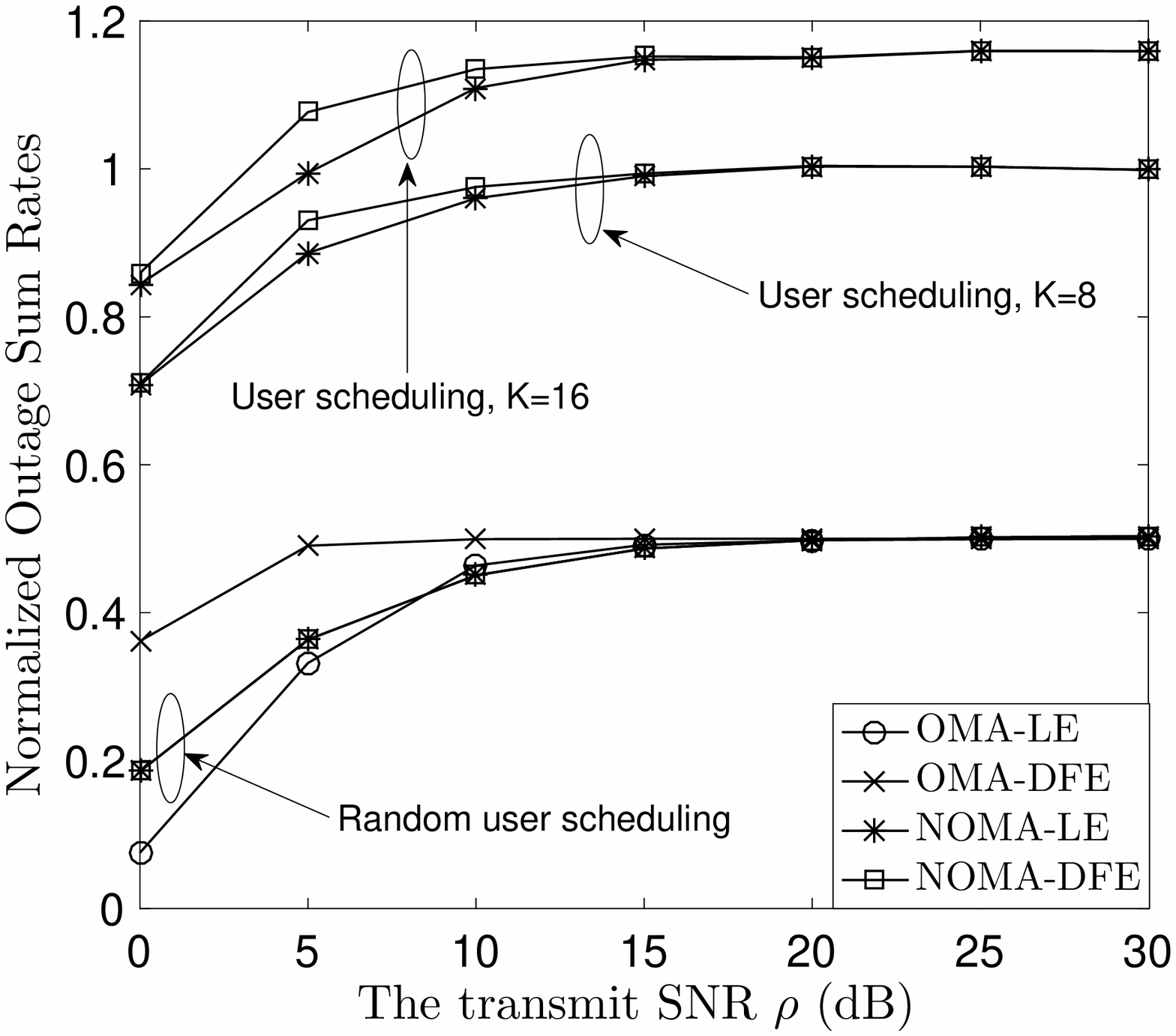}} \subfigure[ Outage Probability]{\label{fig6b}\includegraphics[width=0.43\textwidth]{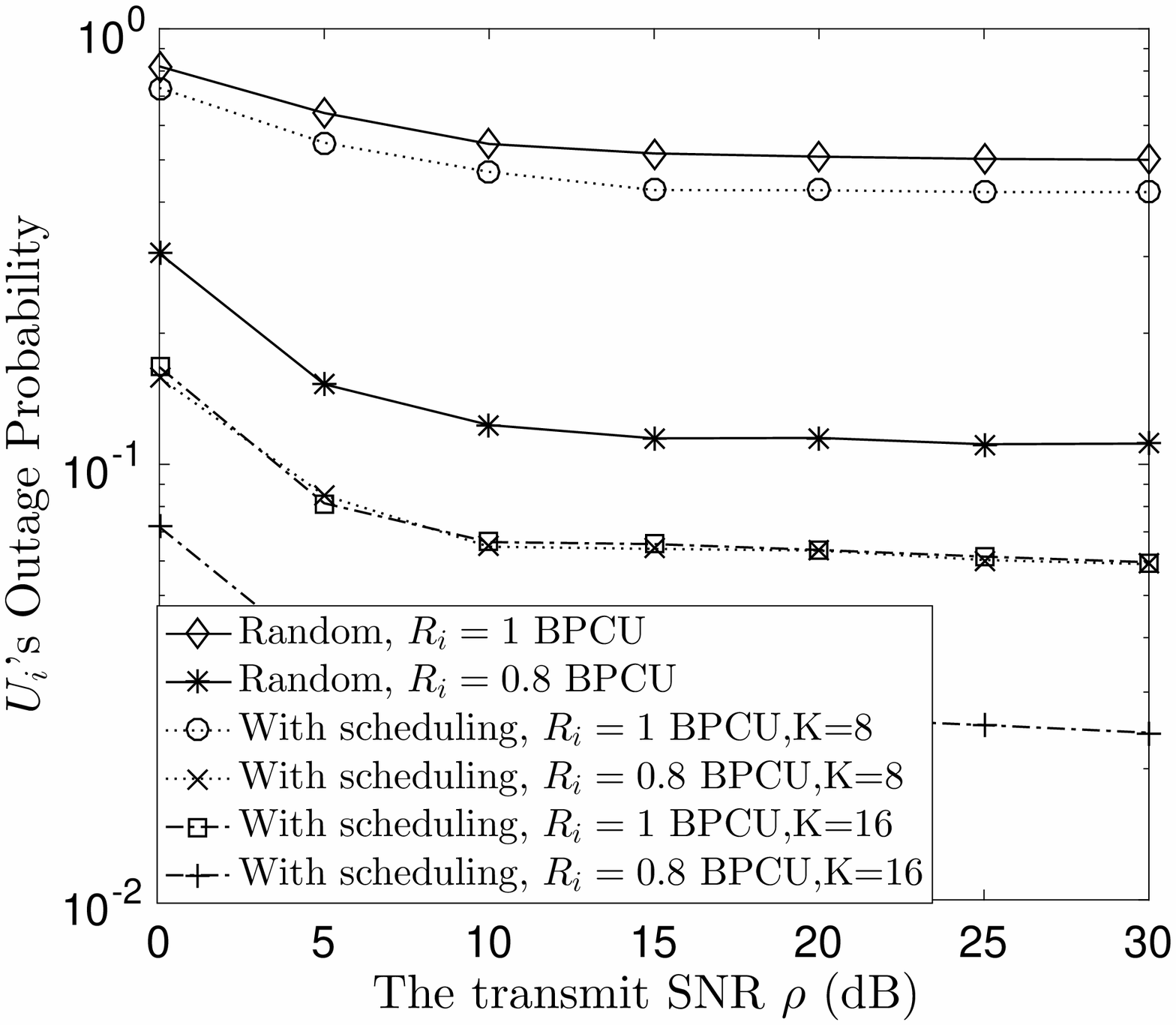}}
\end{center}\vspace{-1em}
 \caption{The performance of uplink   OTFS-NOMA. Fixed-rate transmission is used by the NOMA users. $M=N=16$. $P_0=P_i=3$. $R_0=0.5$  BPCU. $\gamma_0^2=\frac{3}{4}$ and $\gamma_i^2=\frac{1}{4}$ for $i>0$.      }\label{fig6}\vspace{-2em}
\end{figure}

In Figs. \ref{fig5} and \ref{fig6}, the performance of uplink OTFS-NOMA is evaluated. As discussed in Section \ref{section 6}, the NOMA users have two choices for their transmission   rates, namely adaptive and fixed   rate transmission. The use of adaptive rate transmission can ensure that the implementation of NOMA is transparent to $\text{U}_0$, which means that $\text{U}_0$'s QoS requirements  are strictly guaranteed. Since $\text{U}_0$ achieves  the same performance for OTFS-NOMA and OTFS-OMA when adaptive rate transmission  is used, we only focus on the NOMA users' performance, where the ergodic rate   in \eqref{ergodic x} is used as the criterion. We note that this ergodic rate is the net performance gain of OTFS-NOMA over OTFS-OMA, which is the reason why the vertical axis  in Fig. \ref{fig5} is labeled  `Ergodic Rate Gain'. When the $M$ users are randomly selected from the $K$ NOMA users, the ergodic rate gain is moderate, e.g., $1.5$ bit per channel use (BPCU) at $\rho=30$ dB. By applying the   scheduling strategy proposed in \eqref{im selecy2}, the ergodic rate gain can be significantly improved, e.g., nearly by a factor of two compared to the random case   with $K=16$ and $\rho=30$ dB.

Fig. \ref{fig6} focuses on the case with fixed   rate transmission, and similar to Fig. \ref{fig 1}, the normalized outage sum rate is used as   performance criterion in Fig. \ref{fig6a}. One can observe that with random user scheduling, the sum rate of OTFS-NOMA is similar to that of  OTFS-OMA. This is due to the fact that no interference mitigation strategy, such as power or rate allocation,  is used for    NOMA uplink transmission, which means that $\text{U}_0$ and the NOMA users cause strong  interference to each other and SIC failure can happen frequently.  By applying the user scheduling strategy proposed in \eqref{im selecy2}, the   channel conditions of the scheduled users  become quite different, which facilitates the implementation of SIC. This benefit of  user scheduling can be clearly observed  in Fig. \ref{fig6a}, where NOMA achieves   a significant gain over OMA  although advanced power or rate allocation strategies  are not used.    Fig. \ref{fig6a} also shows that the difference between the performance of FD-LE and FD-DFE is insignificant for the uplink case. This is due to the fact that the   outage events during the first stage of SIC dominate the outage performance, and they are not affected by whether  FD-LE or FD-DFE is employed.  Another important observation from Fig. \ref{fig6a} is that the maximal sum rate $R_0+R_i$ cannot be realized, even at high SNR. The reason for this behaviour   is   the existence of the error floor for the NOMA users' outage probabilities, as shown in Fig. \ref{fig6b}. The analytical results provided in   Section \ref{subsection 2} show that increasing $K$ can reduce the error floor, which is   confirmed by Fig. \ref{fig6b}. 

\vspace{-1em}

\section{Conclusions}
In this paper, we have proposed    OTFS-NOMA uplink and downlink transmission schemes, where users with different mobility profiles are grouped together for the implementation of NOMA.  The   analytical results developed in the paper demonstrate that both the high-mobility and low-mobility users benefit from the application  of OTFS-NOMA. In particular, the use of NOMA enables the spreading of the signals of a high-mobility user over a large amount of time-frequency resources, which enhances the OTFS resolution and improves the detection reliability.  In addition, OTFS-NOMA ensures that the low-mobility users have access to the bandwidth resources which would be solely occupied by the high-mobility users in OTFS-OMA. Hence, OTFS-NOMA improves the spectral efficiency and reduces latency.     As shown in the paper, the effective channel gains experienced by  different symbols are different if FD-DFE is employed, which suggests  that    data rate allocation policies  can have a significant impact on the performance of OTFS-NOMA. Therefore,  the design of such policies is an important topic for future research. Another interesting topic for future works is studying the impact of non-zero fractional delays  and fractional Doppler shifts on the performance of the developed OTFS-NOMA protocol. \vspace{-2em}
\appendices
\section{Proof for Proposition \ref{proposition1}} \label{proof1}
Intuitively, the use of $\mathbf{F}_N\otimes \mathbf{F}_M^H$ is analogous to the application of  the  ISFFT which transforms  signals from the delay-Doppler plane to the time-frequency plane, where inter-symbol interference is removed, i.e., the user's channel matrix  is diagonalized. The following proof  confirms this intuition and reveals how the diagonalized channel matrix is related to the original block circulant matrix. We first apply $\mathbf{F}_N\otimes  \mathbf{I}_M$ to $\mathbf{y}_0 $, which yields the following: 
\begin{align}\label{yx model3}
&(\mathbf{F}_N\otimes  \mathbf{I}_M)\mathbf{y}_0= (\mathbf{F}_N\otimes  \mathbf{I}_M)\mathbf{H}_0\left(\gamma_0\mathbf{x}_0 +\sum^{M}_{q=1}\gamma_q {\mathbf{x}}_q \right) + (\mathbf{F}_N\otimes  \mathbf{I}_M)\mathbf{z}_0 \\\nonumber
=&  \diag \left\{\sum_{n=0}^{N-1}  \mathbf{A}_{0,n} e^{-j\frac{2\pi ln}{N}}, 0\leq l \leq N-1\right\}
  (\mathbf{F}_N\otimes  \mathbf{I}_M )\left(\gamma_0\mathbf{x}_0 +\sum^{M}_{q=1}\gamma_q {\mathbf{x}}_q \right)+(\mathbf{F}_N\otimes  \mathbf{I}_M)\mathbf{z}_0 ,
\end{align}
where $\diag \{\mathbf{B}_1, \cdots,  \mathbf{B}_N\}$ denotes a block-diagonal matrix with  $\mathbf{B}_n$, $1\leq n \leq N$, on its main diagonal.  
Note that $ \sum_{n=0}^{N-1}  \mathbf{A}_{0,n} e^{-j\frac{2\pi ln}{N}}$, $0\leq l \leq N-1$, is a sum of $N$ $M\times M$ circulant matrices, each of which can be further diagonalized by $\mathbf{F}_M$. Therefore, we can apply $\mathbf{I}_N \otimes \mathbf{F}_M^H$ to $(\mathbf{F}_N\otimes  \mathbf{I}_M)\mathbf{y}_0$,  which yields the following: 
\begin{align}\label{yx model4}
(\mathbf{I}_N \otimes \mathbf{F}_M^H)( \mathbf{F}_N\otimes  \mathbf{I}_M)\mathbf{y}_0  &= \diag \left\{\sum_{n=0}^{N-1}  \mathbf{\Lambda}_{0,n} e^{-j\frac{2\pi ln}{N}}, 0\leq l \leq N-1\right\}
\\\nonumber &\times (\mathbf{F}_N\otimes  \mathbf{I}_M)( \mathbf{I}_N \otimes \mathbf{F}_M ^H)\left(\gamma_0\mathbf{x}_0 +\sum^{M}_{q=1}\gamma_q {\mathbf{x}}_q \right)  +(\mathbf{I}_N \otimes \mathbf{F}_M^H)( \mathbf{F}_N\otimes  \mathbf{I}_M)\mathbf{z}_0 ,
\end{align}
where $\mathbf{\Lambda}_{0,n}$ is a diagonal matrix, $\mathbf{\Lambda}_{0,n}=\diag\left\{ \sum_{m=0}^{M-1} a^{m,1}_{0,n}e^{j\frac{2\pi tm}{M}}, 0\leq t\leq M-1\right\}$, and $a^{m,1}_{0,n}$ is the element located in the $m$-th row and first column of $\mathbf{A}_{0,n}$. 

By applying  a property of the Kronecker product, $(\mathbf{A} \otimes \mathbf{B})( \mathbf{C}\otimes  \mathbf{D})= (\mathbf{A} \mathbf{C}) \otimes (\mathbf{B}  \mathbf{D}) $, the received signals  can be simplified as follows:
\begin{align}\label{yx model5}
&(\mathbf{F}_N \otimes \mathbf{F}_M ^H ) \mathbf{y}_0 \\\nonumber= &\underset{\mathbf{D_0}}{\underbrace{\diag \left\{\sum_{n=0}^{N-1}  \mathbf{\Lambda}_{0,n} e^{-j\frac{2\pi ln}{N}}, 0\leq l \leq N-1\right\}}}
  (\mathbf{F}_N\otimes \mathbf{F}_M^H )\left(\gamma_0\mathbf{x}_0 +\sum^{M}_{q=1}\gamma_q {\mathbf{x}}_q \right) +(\mathbf{F}_N\otimes \mathbf{F}_M^H )\mathbf{z}_0 ,
\end{align}
where the $(kM+l+1)$-th element on the main diagonal of $\mathbf{D}_0$ is $D_0^{k,l}$ as defined in the proposition.   The proof for the proposition is complete. 
\vspace{-1em}
\section{Proof for Lemma \ref{lemma1}} \label{prooflemma1}
In order to facilitate the SINR analysis, the system model in \eqref{approach 1 system model 1} is further simplified. 
Define $\tilde{X}[n,m] = \sum^{M}_{i=1}X_i[n,m]$. With the mapping scheme used in \eqref{indooruser2}, the NOMA users' signals are interleaved and orthogonally placed in the time-frequency plane, i.e., $\tilde{X}[n,m]$ is simply $\text{U}_{m+1}$'s $n$-th signal, $x_{m+1}(n)$. 
Denote   the outcome of the SFFT of $\tilde{X}[n,m]$ by $\tilde{x}[k,l]$, which yields the following transform: 
\begin{align}\label{lemmaxxx1}
 \tilde{x}[k,l] =\frac{1}{\sqrt{NM}}  \sum^{N-1}_{n=0} \sum^{M-1}_{m=0}\tilde{X}[n,m] e^{-j2\pi \left(\frac{nk}{N}-\frac{ml}{M}\right)}.
\end{align}
Denote    the $NM\times 1$ vector collecting the $\tilde{x}[k,l]$ by $\tilde{\mathbf{x}}$ and    the $NM\times 1$ vector collecting the $\tilde{X}[n,m]$ by $\breve{\mathbf{x}}$, which means that \eqref{lemmaxxx1} can be rewritten as follows:  
\begin{align}
\tilde{\mathbf{x}} =( \mathbf{F}_N \otimes \mathbf{F}_M^H )\breve{\mathbf{x}}.
\end{align}
Therefore, the   model for the received signals  in \eqref{approach 1 system model 1}  can be re-written as follows: 
\begin{align}\label{approach 1 system model 2}
\breve{\mathbf{y}}_0
 = &  \gamma_0\mathbf{x}_0 + \gamma_1 \tilde{\mathbf{x}}  +\left( \mathbf{F}_N \otimes \mathbf{F}_M^H\right)^{-1}\mathbf{D}_0 ^{-1}\tilde{\mathbf{z}}_i\\\nonumber
  = &  \gamma_0\mathbf{x}_0 +\underset{\text{Interference and noise terms}}{\underbrace{\gamma_1(\mathbf{F}_N \otimes \mathbf{F}_M^H) \breve{\mathbf{x}} +\left( \mathbf{F}_N \otimes \mathbf{F}_M^H\right)^{-1}\mathbf{D}_0 ^{-1}\tilde{\mathbf{z}}_0}},
\end{align} 
where we have used the assumption that  $\gamma_i=\gamma_1$, for $1\leq i \leq N$.  
 Note that the power of the information-bearing signals is simply $\gamma_0^2 \rho$, and therefore, the key step to obtain   the SINR is to find the covariance matrix of the interference-plus-noise term. 
 
 We first show that $\tilde{\mathbf{z}}_0\triangleq (\mathbf{F}_N\otimes \mathbf{F}_M^H) \mathbf{z}_0$ is still a complex Gaussian vector, i.e., $\tilde{\mathbf{z}}_i\sim CN(0, \mathbf{I}_{NM})$. Recall that  $\mathbf{z}_0$ contains $N  M$ i.i.d. complex Gaussian random variables. Furthermore,   $\mathbf{F}_N\otimes \mathbf{F}_M^H  $ is a unitary matrix as shown in the following:
\begin{align}
  (\mathbf{F}_N\otimes \mathbf{F}_M^H)  (\mathbf{F}_N\otimes \mathbf{F}_M^H)^H   \overset{(a)} {=}(\mathbf{F}_N\otimes \mathbf{F}_M^H   )(\mathbf{F}_N^H\otimes \mathbf{F}_M )  \overset{(b)}{=}(\mathbf{F}_N\mathbf{F}_N^H)\otimes( \mathbf{F}_M^H  \mathbf{F}_M )=\mathbf{I}_{NM},
\end{align}
where step (a) follows from the fact that $(\mathbf{A}\otimes \mathbf{B})^H=\mathbf{A}^H\otimes \mathbf{B}^H$ and  step (b) follows from the fact that $(\mathbf{A}\otimes \mathbf{B})(\mathbf{C}\otimes \mathbf{D})=(\mathbf{A}\mathbf{C})\otimes (\mathbf{B}\mathbf{D})$. Therefore, $(\mathbf{F}_N\otimes \mathbf{F}_M^H) \mathbf{z}_0\sim CN(0, \mathbf{I}_{NM})$ given the fact that $\mathbf{z}_0\sim CN(0, \mathbf{I}_{NM})$  and  a unitary transformation of a Gaussian vector is still a Gaussian vector. 

Therefore, the covariance matrix of the interference-plus-noise term is given by
\begin{align}\label{approach 1 system model 3}
\mathbf{C}_{\text{cov}}
 = & \gamma_1^2 \mathcal{E}\left\{(\mathbf{F}_N \otimes \mathbf{F}_M^H )\breve{\mathbf{x}} \breve{\mathbf{x}}^H \left(\mathbf{F}_N \otimes \mathbf{F}_M^H\right)^H
 \right\} \\\nonumber &+ \mathcal{E}\left\{\left( \mathbf{F}_N \otimes \mathbf{F}_M^H\right)^{-1}\mathbf{D}_0 ^{-1}\tilde{\mathbf{z}}_0 \tilde{\mathbf{z}}_0^H\mathbf{D}_0 ^{-H}\left( \mathbf{F}_N \otimes \mathbf{F}_M^H\right)^{-H}\right\}.
 \end{align}
 
 Recall that the $(nM+m+1)$-th element of $\breve{\mathbf{x}}$ is $\tilde{X}[n,m]$ which is equal to $x_{m+1}(n)$. Therefore, the   covariance matrix can be further simplified as follows: 
 \begin{align}\label{approach 1 system model 4}
\mathbf{C}_{\text{cov}}
 =& \gamma_1^2  \rho( \mathbf{F}_N \otimes \mathbf{F}_M^H ) \left(\mathbf{F}_N \otimes \mathbf{F}_M^H\right)^H +  \left( \mathbf{F}_N \otimes \mathbf{F}_M^H\right)^{-1}\mathbf{D}_0 ^{-1} \mathbf{D}_0 ^{-H}\left( \mathbf{F}_N \otimes \mathbf{F}_M^H\right)^{-H} \\\nonumber
  =& \gamma_1^2  \rho \mathbf{I}_{MN}   +  \left( \mathbf{F}_N^H \otimes \mathbf{F}_M\right) \mathbf{D}_0 ^{-1} \mathbf{D}_0 ^{-H}\left( \mathbf{F}_N \otimes \mathbf{F}_M^H\right),
\end{align} 
 where the noise power is assumed to be normalized. 

Following  the same steps as in the proof of Proposition \ref{proposition1}, we learn that, by construction, $\left( \mathbf{F}_N^H \otimes \mathbf{F}_M\right) \mathbf{D}_0 ^{-1} \mathbf{D}_0 ^{-H}\left( \mathbf{F}_N \otimes \mathbf{F}_M^H\right)
$ is also a block-circulant matrix, which means that the   elements on the main diagonal of $\left( \mathbf{F}_N^H \otimes \mathbf{F}_M\right) \mathbf{D}_0 ^{-1} \mathbf{D}_0 ^{-H}\left( \mathbf{F}_N \otimes \mathbf{F}_M^H\right)
$ are identical. Without loss of generality, denote the diagonal elements of  $\left( \mathbf{F}_N^H \otimes \mathbf{F}_M\right) \mathbf{D}_0 ^{-1} \mathbf{D}_0 ^{-H}\left( \mathbf{F}_N \otimes \mathbf{F}_M^H\right)$  by $\phi$.  Therefore, $\phi$ can be found by using the trace of the matrix as follows: 
\begin{align}
\phi = &\frac{1}{NM}\text{Tr}\left\{\left( \mathbf{F}_N^H \otimes \mathbf{F}_M\right) \mathbf{D}_0 ^{-1} \mathbf{D}_0 ^{-H}\left( \mathbf{F}_N \otimes \mathbf{F}_M^H\right)\right\}\\\nonumber =&\frac{1}{NM}\text{Tr}\left\{\left( \mathbf{F}_N \otimes \mathbf{F}_M^H\right)\left( \mathbf{F}_N^H \otimes \mathbf{F}_M\right) \mathbf{D}_0 ^{-1} \mathbf{D}_0 ^{-H}\right\} \\\nonumber=&\frac{1}{NM}\text{Tr}\left\{  \mathbf{D}_0 ^{-1} \mathbf{D}_0 ^{-H}\right\} = \frac{1}{NM}\sum^{N-1}_{k=0}\sum^{M-1}_{l=0} |D_0^{k,l}|^{-2}.
\end{align}
Therefore,  the SINR for detection of  $x_0[k,l]$ is given by
\begin{align}
\text{SINR}_{0,kl}^{LE}  = \frac{\rho\gamma_0^2}{\rho \gamma^2_1 +\phi},
\end{align}
and the proof is complete.
\vspace{-1em}
\section{Proof for Lemma \ref{lemmax2}}\label{plemmax2}

The lemma is proved by first developing   upper and lower bounds on the outage probability, and then showing that both bounds  have the same diversity order. 

An upper bound on $\text{SINR}_{0,kl} $ is given by 
\begin{align}
\text{SINR}_{0,kl} =&  \frac{\rho\gamma_0^2}{\rho \gamma^2_1 + \frac{1}{NM}\sum^{N-1}_{\tilde{k}=0}\sum^{M-1}_{\tilde{l}=0} |D_0^{\tilde{k},\tilde{l}}|^{-2}}\  \leq \frac{\rho\gamma_0^2}{\rho \gamma^2_1 + \frac{1}{NM}  |D_0^{0,0}|^{-2}}.
\end{align}

Therefore, the outage probability, denoted by $\mathrm{P}_{0,kl}$,  can be lower bounded as follows:
\begin{align}\label{lowerboundx}
\mathrm{P}_{0,kl} \geq& \mathrm{P}\left(  \frac{\rho\gamma_0^2}{\rho \gamma^2_1 + \frac{1}{NM}  |D_0^{0,0}|^{-2}}<\epsilon_0\right)  
=\mathrm{P}\left(|D_0^{0,0}|^{2} <\frac{\epsilon_0}{NM\rho ( \gamma_0^2-\gamma_1^2\epsilon_0)} \right),
\end{align}
where we assume that $\gamma_0^2 > \gamma^2_1\epsilon_0$. Otherwise, the outage probability is always one.

To evaluate the lower bound on the outage probability, the distribution of  $D_{0}^{u,v}$ is required. Recall from \eqref{yx model6} that  $D_{0}^{u,v}$ is the   $((v-1)M+u)$-th diagonal element of $\mathbf{D}_0$ and can be expressed as follows:
\begin{align}\label{dmatrix}
D_{0}^{u,v}= \sum^{N-1}_{n=0}\sum^{M-1}_{m=0} a_{0,n}^{m,1} e^{j2\pi \frac{um}{M}} e^{-j2\pi \frac{vn}{N}},
\end{align}
which is   the ISFFT of $a_{0,n}^{m,1}$.
Therefore, we  have the following property:  
\begin{align}\label{dmatrix3}
\tilde{\mathbf{D}}_{0}= \sqrt{NM} \mathbf{F}_M ^H\mathbf{A}_{0}  \mathbf{F}_N,
\end{align}
where the element in the $u$-th row and the $v$-th column of $\tilde{\mathbf{D}}_{0}$ is $D_{0}^{u,v}$ and the element in the $m$-th row and the $n$-th column of $\mathbf{A}_{0} $ is   $a_{0,n}^{m,1}$. 

The matrix-based expression shown in \eqref{dmatrix3} can be vectorized as follows:
\begin{align} \label{channel rela}
\rm{Diag} (\mathbf{D}_0) =&  \text{vec} (\tilde{\mathbf{D}}_0) =\sqrt{NM}\text{vec} (\mathbf{F}_M^H \mathbf{A}_{0}  \mathbf{F}_N) 
 =\sqrt{NM} ( \mathbf{F}_N\otimes \mathbf{F}_M^H )\text{vec}(\mathbf{A}_{0} ),
\end{align}
where $\rm{Diag} (\mathbf{A})$ denotes a vector collecting all   elements on the main diagonal of $\mathbf{A}$ and we use the facts that $(\mathbf{C}^T\otimes \mathbf{A})\text{vec}(\mathbf{B})=\text{vec}(\mathbf{D})$ if $\mathbf{A}\mathbf{B}\mathbf{C}=\mathbf{D}$, and $\mathbf{F}_N^T=\mathbf{F}_N$. 

We note that $\text{vec}(\mathbf{A}_{0} )$   contains only $(P_0+1)$ non-zero elements, where the remaining elements are zero.  Therefore, each element on the main diagonal of $\mathbf{D}_0$ is a superposition  of $(P_0+1)$  i.i.d.   random variables, $h_{i,p}\sim CN\left(0,\frac{1}{P_0+1}\right)$. We further note that   the   coefficients for the superposition are complex exponential constants, i.e., the magnitude of each   coefficient is one.  Therefore,  each element on the main diagonal of $\mathbf{D}_0$ is still complex Gaussian distributed, i.e., $D_{0}^{u,v} \sim CN(0,1)$, which means that the lower bound on the outage probability shown in \eqref{lowerboundx} can be expressed  as follows:
\begin{align}
\mathrm{P}_{0,kl} \geq& 1 - e^{- \frac{\epsilon_0}{ NM\rho( \gamma_0^2-\gamma_1^2\epsilon_0) } } \doteq \frac{1}{\rho}.
\end{align}

On the other hand, an upper bound on the outage probability is given by
\begin{align}
\mathrm{P}_{0,kl} \leq& \mathrm{P}\left(  \frac{\rho\gamma_0^2}{\rho \gamma^2_1 + \frac{1}{NM}\sum^{N-1}_{\tilde{k}=0}\sum^{M-1}_{\tilde{l}=0} |D_0^{\text{min}}|^{-2}}<\epsilon_0\right), 
\end{align}
where $|D_0^{\text{min}}| = \min \{|D_0^{{k,l}}|, \forall l \in\{0, \cdots, M-1\}, k\in \{0, \cdots, N-1\} \}$. 

Therefore, the outage probability can be upper bounded as follows: 
\begin{align}
\mathrm{P}_{0,kl} \leq& \mathrm{P}\left(|D_0^{\text{min}}|^{2} <\frac{\epsilon_0}{\rho( \gamma_0^2-\gamma_1^2\epsilon_0)}  \right)  .
\end{align}
It is important to point out that  the $|D_0^{k,l}|^{2}$, $ l \in\{0, \cdots, M-1\}, k\in \{0, \cdots, N-1\} $, are identically but not independently distributed.   This correlation property is shown as follows. The covariance  matrix of the effective channel gains, i.e., the elements on the main diagonal of $\mathbf{D}_0$, is given by
\begin{align}
\mathcal{E}\left\{ \rm{Diag}(\mathbf{D}_0)  \rm{Diag}(\mathbf{D}_0) ^H\right\}  =& NM\mathcal{E}\left\{   ( \mathbf{F}_N\otimes \mathbf{F}_M^H )\text{vec}(\mathbf{A}_{0} ) \text{vec}(\mathbf{A}_{0} ) ^H( \mathbf{F}_N\otimes \mathbf{F}_M^H )^H\right\}
\\\nonumber=& NM( \mathbf{F}_N\otimes \mathbf{F}_M^H ) \mathcal{E}\left\{  \text{vec}(\mathbf{A}_{0} ) \text{vec}(\mathbf{A}_{0} ) ^H\right\}( \mathbf{F}_N\otimes \mathbf{F}_M^H )^H.
\end{align}
Because the channel gains, $h_{0,p}$,  are i.i.d., $ \mathcal{E}\left\{  \text{vec}(\mathbf{A}_{0} ) \text{vec}(\mathbf{A}_{0} ) ^H\right\}$ is a diagonal matrix, where only $(P_0+1)$ of its diagonal  elements are non-zero. Following the same steps as in the proof for Proposition \ref{proposition1},   one can show that the product of $ ( \mathbf{F}_N\otimes \mathbf{F}_M^H ) $, a diagonal matrix, and   $( \mathbf{F}_N\otimes \mathbf{F}_M^H )^H$ yields a block circulant matrix, which means that $\mathcal{E}\left\{ \rm{Diag}(\mathbf{D}_0)  \rm{Diag}(\mathbf{D}_0) ^H\right\}$ is a block-circulant matrix,   not a diagonal    matrix. Therefore, the $|D_0^{k,l}|^{2}$, $ l \in\{0, \cdots, M-1\}, k\in \{0, \cdots, N-1\} $ are correlated, instead of independent.  

Although the $|D_0^{k,l}|^{2}$ are not independent, an upper bound on $\mathrm{P}_{0,kl}$ can be still found  as follows:
\begin{align}
\mathrm{P}_{0,kl} \leq& \mathrm{P}\left(|D_0^{\text{min}}|^{2} <\frac{\epsilon_0}{\rho( \gamma_0^2-\gamma_1^2\epsilon_0)}  \right)   \leq  \sum^{N-1}_{k=0} \sum^{M-1}_{l=0} \mathrm{P}\left(|D_0^{k,l}|^{2} <\frac{\epsilon_0}{\rho( \gamma_0^2-\gamma_1^2\epsilon_0)}  \right) \\\nonumber
\leq& MN \mathrm{P}\left(|D_0^{0,0}|^{2} <\frac{\epsilon_0}{\rho( \gamma_0^2-\gamma_1^2\epsilon_0)}  \right)  =MN\left(1 - e^{-\frac{\epsilon_0}{\rho( \gamma_0^2-\gamma_1^2\epsilon_0)} } \right)\doteq \frac{1}{\rho}. 
\end{align}
Since both the upper and lower bounds on the outage probability have the same diversity order, the proof of the lemma is complete.

\vspace{-1em}
\linespread{1.2}
   \bibliographystyle{IEEEtran}
\bibliography{IEEEfull,trasfer}
   \end{document}